\documentclass[10pt, a4paper]{article}
\usepackage{comment}
\usepackage{amsmath}
\usepackage{amssymb}
\usepackage{amsthm}
\usepackage{amscd}
\usepackage{graphicx}
\usepackage{indentfirst}
\usepackage[ruled, vlined, linesnumbered]{algorithm2e}
\usepackage{titlesec}
\usepackage[top=25.4mm,  bottom=25.4mm,  left=31.7mm,  right=31.2mm]{geometry}
\usepackage{titlesec}
\usepackage{color}
   \usepackage{amsmath, amssymb}
   \usepackage{latexsym}
\newcommand{\mycite}[1]{$^{\mbox{\rm\scriptsize\cite{#1}}}\!$}

\begin{document}
\newtheorem{definition}{\bf Definition}
\newtheorem{theorem}{\bf Theorem}
\newtheorem{proposition}{\bf Proposition}
\newtheorem{lemma}{\bf Lemma}
\newtheorem{corollary}{\bf Corollary}
\newtheorem{example}{\bf Example}
\newtheorem{remark}{\bf Remark}
\newtheorem{Table}{\bf Table}
\newtheorem{Sentence}{\bf Step}
\newtheorem{Branch}{}

\def\T {\ensuremath{\bf{T}}}
\def\N {\ensuremath{\mathbb{N}}}
\def\A {\ensuremath{\bf{A}}}
\def\B {\ensuremath{\bf{B}}}
\def\P {\ensuremath{\bf{P}}}
\def\S {\ensuremath{\mathbb{S}}}
\def\E {\ensuremath{\bf{E}}}
\def\H {\ensuremath{\bf{H}}}
\def\V {\ensuremath{\rm{V}}}
\def\D {\ensuremath{\rm{D}}}
\def\PF {\ensuremath{\bf {PF}}}
\def\TH {\ensuremath{\bf{TH}}}
\def\RS {\ensuremath{\mathbb{T}}}
\def\HCTD {\ensuremath{\tt{HCTD}}}
\def\HPCTD {\ensuremath{\mathrm{HPCTD}}}
\def\CTD {\ensuremath{\mathrm{CTD}}}
\def\PCTD {\ensuremath{\mathrm{PCTD}}}
\def\WUCTD {\ensuremath{\mathrm{WUCTD}}}
\def\RSD {\ensuremath{\mathrm{RSD}}}
\def\FWCTD {\ensuremath{\mathrm{FWCTD}}}
\def\SWCTD {\ensuremath{\mathrm{SWCTD}}}
\def\ARSD {\ensuremath{\tt{RSD}}}
\def\APRSD {\ensuremath{\tt{WRSD}}}
\def\CCTD {\ensuremath{\tt{CTD}}}
\def\ASWCTD {\ensuremath{\tt{SWCTD}}}
\def\ASMPD {\ensuremath{\tt{SMPD}}}
\def\AHPCTD {\ensuremath{\tt{HPCTD}}}
\def\TDU {\ensuremath{\mathrm{TDU}}}
\def\ATDU {\ensuremath{\tt{TDU}}}
\def\RDU {\ensuremath{\mathrm{RDUForZD}}}
\def\ARDU {\ensuremath{\tt{RDUForZD}}}
\newtheorem{Rules}{\bf Rule}

\newcommand{\disc}[1]{\mbox{{\rm disc}$(#1)$}}
\newcommand{\alg}[1]{\mbox{{\rm alg}$(#1)$}}
\newcommand{\SAT}[1]{\mbox{{\rm sat}$(#1)$}}
\newcommand{\SQR}[1]{\mbox{{\rm sqrt}$(#1)$}}
\newcommand{\ideal}[1]{\langle#1\rangle}
\newcommand{\I}[1]{\mbox{{\rm I}$_{#1}$}}
\newcommand{\ldeg}[1]{\mbox{{\rm ldeg}$(#1)$}}
\newcommand{\iter}[1]{\mbox{{\rm iter}$(#1)$}}
\newcommand{\mdeg}[1]{\mbox{{\rm mdeg}$(#1)$}}
\newcommand{\lv}[1]{\mbox{{\rm lv}$_{#1}$}}
\newcommand{\mvar}[1]{\mbox{{\rm mvar}$(#1)$}}
\newcommand{\prem}[1]{\mbox{{\rm prem}$(#1)$}}
\newcommand{\pquo}[1]{\mbox{{\rm pquo}$(#1)$}}
\newcommand{\rank}[1]{\mbox{{\rm rank}$(#1)$}}
\newcommand{\res}[1]{\mbox{{\rm res}$(#1)$}}
\newcommand{\cls}[1]{\mbox{{\rm cls}$_{#1}$}}
\newcommand{\sat}[1]{\mbox{{\rm sat}$(#1)$}}
\newcommand{\sep}[1]{\mbox{{\rm sep}$(#1)$}}
\newcommand{\tail}[1]{\mbox{{\rm tail}$(#1)$}}
\newcommand{\zm}[1]{\mbox{{\rm MZero}$(#1)$}}
\newcommand{\zero}[1]{\mbox{{\rm Zero}$(#1)$}}
\newcommand{\rd}[1]{\mbox{{\rm Red}$(#1)$}}
\newcommand{\map}[1]{\mbox{{\rm map}$(#1)$}}
\newcommand{\op}[1]{\mbox{{\rm op}$(#1)$}}

\title{ {Generic Regular Decompositions for Generic Zero-Dimensional Systems}}
\author{Xiaoxian Tang\thanks{Corresponding author}\hspace{1.0em}Zhenghong Chen\hspace{1.0em}Bican Xia\\
         {\small LMAM \& School of Mathematical Sciences}\\
         {\small Peking University,  Beijing 100871,  China}\\
         {\small tangxiaoxian@pku.edu.cn, \ \ chenzhenghong@pku.edu.cn, \ \ xbc@math.pku.edu.cn}}
\date{}
\maketitle
\begin{abstract}
  Two new concepts,   generic regular decomposition and regular-decomposition-unstable (RDU) variety for generic zero-dimensional  systems,  are introduced in this paper and an algorithm is proposed for computing a generic regular decomposition and the associated RDU variety of a given generic zero-dimensional system simultaneously. The solutions of the given system can be expressed by finitely many zero-dimensional regular chains if the parameter value is not on the RDU variety.
 The so called weakly relatively simplicial decomposition plays a crucial role in the algorithm,  which is based on the theories of subresultants. Furthermore,  the algorithm can be naturally adopted to compute a non-redundant Wu's decomposition and the decomposition is stable at any parameter value that is not on the RDU variety. The algorithm has been implemented with Maple 16 and experimented with a number of benchmarks from the literature. Empirical results are also presented to show the good performance of the algorithm.

~
\end{abstract}

\section{Introduction}\label{SecIntro}
Solving parametric polynomial systems is usually a key problem in many research and applied areas,  such as automated geometry theorem deduction,  stability analysis of dynamical systems,  robotics and so on \cite{changbo, Montes, CGS}. By ``solving",  we often mean to determine (1) for what parameter values the polynomial system has solutions,  and (2) whether the solutions can be expressed by some simple representations.

Generally speaking,  there are two kinds of methods for solving the above questions (1) and (2),  {\it i.e.},  the methods based on {\it Gr\"obner bases}  \cite{sun, Montes, KN, CGS} and {\it triangular decompositions} \cite{marco, changbo, gxs1992, kalk, maza, dkwang, wangi, wu, yhx01, xia, zjzi}.

For parametric systems, the concepts of comprehensive Gr\"obner system (CGS) and comprehensive Gr\"obner bases (CGB) introduced by Weispfenning in \cite{CGS} and the algorithms for computing them  \cite{sun,Montes,KN,SS,newSS,CGS} are powerful tools for answering questions (1) and (2). The first CGB algorithm introduced in\cite{CGS} suffers from the problem of too many redundant branches. Many improved algorithms have been proposed since then \cite{sun,Montes,KN,SS,newSS}, among which, the one proposed by Suzuki and Sato \cite{SS} was accepted widely by subsequent researchers. The latest progress on this subject was reported by Kapur {\it et al.} \cite{sun}. They solved the famous P3P problem\cite{gxs} by computing CGS and provided empirical data illustrating that the CGS method could solve practical problems in amazingly short time.

The methods based on triangular decompositions have been studied by many researchers since Wu's work  \cite{wu}.  A significant concept in the theories of triangular sets is ``regular chain" (or ``normal chain") introduced by Kalkbrener \cite{kalk} and Yang and Zhang \cite{zjzi}  independently. Gao and Chou proposed a method in \cite{gxs1992}  for identifying all parametric values for which a given system has solutions and giving the solutions by $p-$chains\footnotemark \footnotetext{The concept $p-$chain is stronger than regular chain,  see more details in \cite{gxs1992}.} without a partition of the parameter space.
Wang generalized the concept of regular chain to regular system and gave an efficient algorithm for computing it \cite{wangi,  wang, wangEpsilon}. It should be noticed that,  due to their strong projection property,  the regular systems or series computed by {\tt RegSer}\footnotemark\footnotetext{http://www-calfor.lip6.fr/\~{}wang/epsilon/} may also be used as representations for parametric systems.
The concept of comprehensive triangular decomposition (CTD) introduced by Chen {\it et al.} in \cite{changbo}  can answer questions (1) and (2). Algorithms for computing regular chain decompositions and CTDs have been implemented as central functions of {\tt RegularChains} library in Maple 16.

For a given parametric system $\P$ with $n$ variables and $d$ parameters, many existing algorithms for computing regular decomposition over a certain field $K$ give a regular zero-decomposition of $\P$ in ${\overline K}^{n+d}$. Then, if one wants to answer questions (1) and (2), one may try computing projections from the solution space  to the parametric space. On the other hand, there are some other methods, such as Wu's method \cite{wu} and relatively simplicial decomposition (RSD) \cite{zjzi}, which consider parameters as ``constants" during the process of decomposition and can obtain zero-decompositions of $\P$ in $\overline{K(U)}^{n}$ where $U$ stands for the $d$ parameters. In this paper, we follow the idea of the latter methods and propose an algorithm for computing a so-called {\em generic regular decomposition} ${\mathbb T}$
of a generic zero-dimensional system $\P$ in $\overline{K(U)}^{n}$ (see Definition \ref{DEsus}). 
At the same time, the algorithm also obtains a parametric polynomial such that the regular decomposition is {\em stable} at any parametric point outside the variety generated by the parametric polynomial and we call the variety {\em regular-decomposition-unstable} (RDU) variety. Roughly speaking, ``stable at a parametric point" means that the regular decomposition will remain after we substitute the point for the parameters in $\P$ and ${\mathbb T}$ (see Definition \ref{DEsus}). As a result,  questions (1) and (2) for generic zero-dimensional systems are answered except for the case where parameters are on the RDU variety. That is why the decomposition is called {\em generic} regular chain decomposition.

The proposed algorithm is based on {\em weakly relatively simplicial decomposition},  a new concept that is weaker than {\em relatively simplicial decomposition}  proposed by Yang {\it et al.} in \cite{zjzi}  and inspired by the method for computing regular systems introduced by Wang in \cite{wangi, wang}.  In addition,  the proposed algorithm can be naturally adopted to compute a non-redundant Wu's decomposition for a given generic zero-dimensional system. Furthermore, computing RDU varieties can be regarded as the first step of computing {\em border polynomial} (BP), which is a crucial concept introduced by Yang {\it et al.} \cite{yhx01,  yang,  xia} for solving the real root classification (RRC) problem of parametric semi-algebraic systems.  As a matter of fact, an RDU variety of a generic zero-dimensional system {\rm with respect to} ({\it w.r.t.}) a generic regular decomposition  is a subvariety of the hypersurface generated by a certain BP.
The new algorithm has been implemented on the basis of DISCOVERER \cite{discover}  with Maple 16 and experimented with a number of benchmarks from \cite{changbo,  zxq, sun,  Montes,  KN}.  Empirical results are also presented to show the good performance of the algorithm.

The paper is organized as follows. Section \ref{sfuhaoshuoming} gives basic definitions and concepts that are needed to understand the main algorithm. Section \ref{sectionus} contains the main algorithm,  namely Algorithm \ref{ALsus},  and some relative subalgorithms,  especially the subalgorithm for computing weakly relatively simplicial decompositions. Besides,  proofs for these algorithms are presented in this section and several illustrative examples are given. The empirical data and comparison with previous work along with several implementation details are presented in Section \ref{sectionexamples}.  Section 5 concludes the paper with a discussion on our future work along this direction.
\section{Preliminaries}\label{sfuhaoshuoming}

All concepts in this section without precise definitions can be found in \cite{cox, wu, xia}.
$\mathbb R$ and $\mathbb C$ stand for the field of real numbers and the field of complex numbers,  respectively.

         Suppose $\{u_1,  \ldots,  u_d,  x_1,  \ldots,  x_n\}$ is a set of indeterminates with a given order $u_1\prec. . . \prec u_d\prec x_1\prec. . . \prec x_n$
        where $\{u_1,  \ldots,  u_d\}$ and $\{x_1,  \ldots,  x_n\}$ are the sets of parameters and variables,  respectively.
       Let $U=\{u_1,  \ldots,  u_d\}$ and $X=\{x_1,  \ldots,  x_n\}$.
       Suppose $K$ is a field and
        $\overline {K}$ is its algebraic closure. Let $K[U]$ be the ring of polynomials in $U$ with coefficients in $K$ and $K(U)$ be the rational function field.
       A non-empty finite subset $\P$ of $K[U][X]$ is said to be a {\em system}. If ${\P}\subset K[U][X]\backslash K[X]$,  it is called a {\em parametric system}. If ${\P}\subset K[X]$,  it is called a {\em constant system}.

        For a system $\P\subset$$K[U][X]$ ($\overline{K}[X]$),  $\ideal{{\P}}_{K[U][X]}$ ($\ideal{{\P}}_{\overline{K}[X]}$) denotes the ideal generated by ${\P}$ in $K[U][X]$ ($\overline{K}[X]$).
        For any $F$ in $K[U][X]\backslash \{0\}$ ($\overline{K}[X]\backslash \{0\}$) and for any $x\in X$,  if $x$ appears in $F$,   $F$ can be regarded as a univariate polynomial in $x$,  namely $F=C_0x^m+C_1x^{m-1}+\ldots+C_m$ where $C_0, C_1, \ldots, C_m$ are polynomials in $K[U][X\backslash \{x\}]$ ($\overline{K}[X\backslash \{x\}]$) and $C_0\neq 0$. Then $m$ is the {\em leading degree} of $F$ {\it w.r.t.} $x$ and is denoted by $\deg(F, x)$. Note that if $x$ does not appear in $F$,  $\deg(F, x)=0$. The class of $F$ is the biggest index $k$ such that $\deg(F,x_k)>0$. If $\deg(F, x_i)=0$ for every $i$ ($1\leq i\leq n$),   then the {\em class} of $F$ is $0$. The class of $F$ in $K[U][X]\backslash \{0\}$ ($\overline{K}[X]\backslash \{0\}$) is denoted by $\cls{F}$. If $\cls{F}>0$,  $x_{\cls{F}}$ is the {\em main variable} of $F$ and is denoted by $\mvar{F}$. Assume that $F=C_0x_p^m+C_1x_p^{m-1}+\ldots+C_m$ where $p=\cls{F}>0$ and $C_0\neq 0$,  then $C_0$, denoted by $\I{F}$, is the {\em initial} of $F$ and $x_p^m$, denoted by $\rank{F}$, is the {\em rank} of $F$.

          A non-empty finite set ${\T}=\{T_1, \ldots, T_r\}$ of polynomials in $K[U][X]$ ($\overline{K}[X]$)  is  a {\em triangular set} in $K[U][X]$ ($\overline{K}[X]$) if $0<\cls{T_1}<\cls{T_2}<\ldots<\cls{T_r}$. For a triangular set $\T$ in $K[U][X]$ ($\overline{K}[X]$),  $\I{{\T}}$,  $\mvar{{\T}}$  and $\rank{{\T}}$ denote $\Pi_{T\in {\T}}\I{T}$,  $\{\mvar{T}|T\in {\T}\}$ and $\{\rank{T}|T\in {\T}\}$, respectively. The {\em saturated ideal} of a triangular set $\T$ in $K[U][X]$ is defined as the set $\{F\in K[U][X]|\I{\T}^sF\in \ideal{{\T}}_{K[U][X]}$ for some positive integer $s\}$ and is denoted by $\SAT{\T}_{K[U][X]}$. Similarly,  the {\em saturated ideal} of a triangular set $\T$ in  $\overline{K}[X]$ is defined as the set $\{F\in \overline{K}[X]|\I{\T}^sF\in \ideal{{\T}}_{\overline{K}[X]}$ for some positive integer $s\}$ and is denoted by $\SAT{\T}_{\overline{K}[X]}$.
 Suppose $F\in K[U][X]$ ($\overline{K}[X]$) and $\T$ is a triangular set in $K[U][X]$  ($\overline{K}[X]$),  then $F$ is {\em reduced} {\it w.r.t.} $\T$ if $\deg(F, \mvar{T_i})<\deg(T_i, \mvar{T_i})$ for every $i$ $(1\leq i\leq r)$. A triangular set ${\T}=\{T_1, \ldots, T_r\}$ in $K[U][X]$ ($\overline{K}[X]$) is a {\em non-contradictory ascending chain} in $K[U][X]$ ($\overline{K}[X]$) if $T_i$ is reduced {\it w.r.t.} $\{T_1, \ldots, T_{i-1}\}$ for every $i$ $(2\leq i\leq r)$. A single-element set $\{F\}\subset K[U]$ ($\{F\}\subset \overline{K}$) is  a {\em contradictory ascending chain} in $K[U][X]$ ($\overline{K}[X]$) if $F\neq 0$. An {\em ascending chain} is either a non-contradictory ascending chain or a contradictory ascending chain.

          For two polynomials $F$ and $P$ in $K[U][X]$ ($\overline{K}[X]$) and a variable $x\in X$,  the {\em pseudo remainder} and the {\em pseudo quotient} of $F$ {\em pseudo-divided} by $P$ {\it w.r.t.} $x$ are
        denoted by $\prem{F,  P,  x}$ and $\pquo{F, P, x}$, respectively. Particularly, $\prem{F, P, \mvar{P}}$ is denoted by $\prem{F, P}$.
        For a polynomial $F\in K[U][X]$ ($\overline{K}[X]$) and a triangular set ${\T}=\{T_1,  . . . ,  T_r\}$ in $K[U][X]$ ($\overline{K}[X]$),
        the {\em successive pseudo remainder} \cite{zjzi} of $F$ {\it w.r.t.} ${\T}$
        is denoted by $\prem{F,  {\T}}$,  namely
$\prem{F,  {\T}}=\prem{\ldots\prem{\prem{F,T_r},T_{r-1}},\ldots,T_1}$.
        For a finite set ${\P}\subset K[U][X]$ ($\overline{K}[X]$),  $\prem{{\P},  {\T}}$
        denotes the set $\{\prem{F,  {\T}}\mid F\in {\P}\}$.

For ${\bf P}\subset K[U][X]$,
       the set
$\{(a_1,  \ldots,  a_n)\in \overline{K(U)}^{n}|P(U, a_1,  \ldots,  a_n)=0,  \forall P\in \P\}$ is denoted by ${\rm V}_{\overline{K(U)}}({\bf P})$.
        An ascending chain ${\bf C}$ in $K[U][X]$  is a {\em characteristic set} of $\P$ in $K[U][X]$ if ${\bf C}\subset \ideal{{\P}}_{K[U][X]}$ and $\prem{{\P}, {\bf C}}=\{0\}$.  Theorem \ref{wellorder} below is the so-called {\em well-ordering principle}.

\begin{theorem}\label{wellorder}\mycite{wu}
There exists an algorithm which, for an input non-empty finite subset ${\P}\subset K[U][X]$, outputs either a contradictory ascending chain meaning that ${\V}_{\overline{K(U)}}({\P})=\emptyset$, or a (non-contradictory) characteristic set ${\bf C}=\{C_1,\ldots,C_t\}$ such that
${\V}_{\overline{K(U)}}({\P})={\V}_{\overline{K(U)}}({\bf C}\backslash\I{{\bf C}})\cup\cup_{i=1}^t{\V}_{\overline{K(U)}}({\P}\cup {\bf C}\cup \{\I{C_i}\})$.
\end{theorem}

On the base of Theorem \ref{wellorder}, there exists an algorithm, namely Wu's method,  for computing a  finite sequence of ascending chains ${\bf C}_1, {\bf C}_2, \ldots, {\bf C}_m$ $(m\geq 1)$ in $K[U][X]$ such that ${\bf C}_1, {\bf C}_2, \ldots, {\bf C}_m$ is a finite sequence of  characteristic sets  in $K[U][X]$ and
if $m=1$,  ${\V}_{\overline{K(U)}}({\P})=\emptyset$; otherwise,  suppose ${\mathbb S}=\{{\bf C}_i|1\leq i\leq m$ and ${\bf C}_i$ is a non-contradictory ascending chain$\}$. Then  ${\V}_{\overline{K(U)}}({\P})=\cup_{{\bf C}\in {\mathbb S}}{\V}_{\overline{K(U)}}({\bf C}\backslash\I{{\bf C}})$.

 The set of ascending chains $\{{\bf C}_1, {\bf C}_2, \ldots, {\bf C}_m\}$ above is said to be a {\em Wu's decomposition} or {\em characteristic set decomposition} of $\P$ in $K[U][X]$.
In addition, $\P$ is said to be a {\em generic zero-dimensional system} if $\mvar{{\bf C}_i}=X$ for every non-contradictory ascending chain ${\bf C}_i$.  Remark that a Wu's decomposition may suffer from the redundant branches problem. That means,  ${\V}_{\overline{K(U)}}({\bf C}_i\backslash\I{{\bf C}_i})$ can be an empty set for some non-contradictory ascending chain ${\bf C}_i$ $(1\leq i\leq m)$.

 Another important concept in the theories of triangular decompositions is regular chain.
        For two polynomials $F$ and $P$ in $K[U][X]$ ($\overline{K}[X]$) and a variable $x\in X$,  the {\em resultant} [15] of $F$ and $P$ {\it w.r.t.} $x$ is
        denoted by $\res{F,  P,  x}$. Particularly,  $\res{F, P, \mvar{P}}$ is denoted by $\res{F, P}$.
        For a polynomial $F\in K[U][X]$ ($\overline{K}[X]$) and a triangular set ${\T}=\{T_1,  . . . ,  T_r\}$ in $K[U][X]$ ($\overline{K}[X]$),
        the {\em successive resultant} \cite{zjzi} of $F$ {\it w.r.t.} ${\T}$
        is denoted by $\res{F,  {\T}}$,  namely
$\res{F,  {\T}}=\res{\ldots\res{\res{F,T_r},T_{r-1}},\ldots,T_1}$.
       A triangular set ${\bf T}=\{T_1,  \ldots,  T_r\}$ in
       $K[U][X]$ ($\overline{K}[X]$) is said to be a {\em regular chain} in $K[U][X]$ ($\overline{K}[X]$),   if $\I{{T_1}}\neq 0$ and for each $i$ $(1<i\leq r)$,  $\res{\I{T_i},  \{T_{i-1},  \ldots,  T_1\}}\neq 0$.
       If ${\bf T}$ is a regular chain in $K[U][X]$ ($\overline{K}[X]$) and $\mvar{{\T}}=X$,  $\T$ is a {\em zero-dimensional regular chain}.
       Regular chains have a series of good properties,  some of which are listed below.  For ${\bf P}\subset \overline{K}[X]$,
       ${\V}({\bf P})$ denotes the set
$\{(a_1,  \ldots,  a_n)\in \overline{K}^{n}|P(a_1,  \ldots,  a_n)=0,  \forall P\in \P\}$.

\begin{proposition}\label{PRRCP1}\mycite{marco, changbo, kalk, wangi, wang, zjzi, zjziii}
If $\T$ is a regular chain in $K[U][X]$ $(\overline{K}[X])$,  then ${\V}_{\overline{K(U)}}(\T\backslash\I{\T})\neq \emptyset$ $({\V}(\T\backslash \I{\T})\neq \emptyset)$.
\end{proposition}
\begin{proposition}\label{PRRCP2}\mycite{marco, changbo, kalk, wangi, wang, zjzi, zjziii}
If $\T$ is a regular chain in $K[U][X]$ and $P$ is a polynomial in $K[U][X]$,  then

$(1)$$\prem{P, {\T}}=0$ if and only if $P\in \SAT{\T}_{K[U][X]}$;

$(2)$${\V}_{\overline{K(U)}}({\T}\backslash \I{\T})\subset {\V}_{\overline{K(U)}}(P)$ if and only if $P\in \sqrt{\SAT{\T}_{K[U][X]}}$.\\
Furthermore, if $\T$ is zero-dimensional, then

$(3)$${\V}_{\overline{K(U)}}({\T})\cap {\V}_{\overline{K(U)}}(P)\neq \emptyset$ if and only if $\res{P, {\T}}= 0$.
\end{proposition}

\begin{remark}\label{REPRRCP2}\mycite{marco, changbo, kalk, wangi, wang, zjzi, zjziii}
This remark is an analogue of Proposition \ref{PRRCP2}. If $\T$ is a regular chain in $\overline{K}[X]$ and $P$ is a polynomial in $\overline{K}[X]$,  then

$(1)$$\prem{P, {\T}}=0$ if and only if $P\in \SAT{\T}_{\overline{K}[X]}$;

$(2)$${\V}({\T}\backslash \I{\T})\subset {\V}(P)$ if and only if $P\in \sqrt{\SAT{\T}_{\overline{K}[X]}}$.\\
Furthermore, if $\T$ is zero-dimensional, then

$(3)$${\V}({\T})\cap {\V}(P)\neq \emptyset$ if and only if $\res{P, {\T}}= 0$.
\end{remark}

\begin{remark}\label{REidea}
     There exist various efficient algorithms for computing regular chain decompositions \cite{marco, changbo, kalk, wangi, wang, zjzi, zjziii}. Regular chain decompositions do not suffer from redundant problem as Wu's decompositions owing to  Proposition \ref{PRRCP1}.  It should be noted that the definition of triangular set and thus that of regular chain in $K[U][X]$ introduced above is not exactly the same as that introduced in \cite{changboPHD, changbo, wangi} when dealing with parametric systems. For example,  consider a parameter system $\{u, x_1, x_2\}$ in $\mathbb{R}[u][x_1, x_2]$. The system itself is a regular chain in $\mathbb{R}[u][x_1, x_2]=\mathbb{R}[u, x_1, x_2]$ according to the definition of regular chain introduced in \cite{changboPHD, changbo, gxs1992, wangi}.  But $\{u, x_1, x_2\}$ is not a regular chain in $\mathbb{R}[u][x_1, x_2]$ in this paper.
\end{remark}

\begin{definition}\label{revise3}
      Suppose $\P$ is a generic zero-dimensional system in $K[U][X]$. A finite set  $\mathbb{T}$  of triangular sets in $K[U][X]$ is said to be a {\em parametric triangular decomposition} of ${\bf P}$ in $K[U][X]$  if ${\rm V}_{\overline{K(U)}}({\bf P})=\cup _{{\T}\in {\mathbb T}}{\rm V}_{\overline{K(U)}}({\bf T}\backslash \I{{\T}})$. If $\mathbb{T}=\emptyset$ or ${\V}_{\overline{K(U)}}({\bf T}\backslash\I{{\bf T}})\neq \emptyset$ for any $\T\in \mathbb{T}$,   the parametric triangular decomposition is said to be {\em non-redundant}. If ${\mathbb T}$ is a finite set of regular chains in $K[U][X]$,  the parametric triangular decomposition is said to be a {\em parametric regular decomposition}.
\end{definition}

       For each $a=(a_1,  \ldots,  a_d)\in {\overline{K}}^{d}$,   $\phi
        _{a}:K[U][X]\longrightarrow \overline{K}[X]$ is a homomorphism such  that $\phi_{a}(F)=F(a,  X)$ for all
        $F\in K[U][X]$ and we denote $\phi_{a}(F)$ by $F(a)$. For a non-empty finite set ${\P}\subset K[U][X]$,  ${\P}(a)$ denotes the set $\{F(a)|F\in{\P}\}$ and   ${\P}(a)=\emptyset$ if ${\P}=\emptyset$.

\begin{definition}\label{DEfus}
      Let $\mathbb{T}$ be a parametric triangular decomposition of a given generic zero-dimensional system $\P$ in $K[U][X]$.  $\mathbb{T}$ is said to be {\em stable} at $a\in \overline{K}^d$  if $\V({\P}(a))=\cup_{{\T}\in \mathbb{T}}\V({\T}(a)\backslash \I{{\T}(a)})$ and $\rank{{\T}}=\rank{{\T}(a)}$ for any ${\T}\in \mathbb{T}$.
\end{definition}

\begin{definition}\label{DeRSSwell}\mycite{changbo}
     Let ${\bf T}$ be a regular chain in $K[U][X]$ and $a\in\overline{K}^{d}$.   If ${\bf
     T}(a)$ is a regular chain in $\overline{K}[X]$ and
     $\rank{{\bf T}(a)}$$=\rank{{\bf T}}$,
     then we say that the regular chain ${\bf T}$ {\em specializes well} at $a$.
\end{definition}

Suppose $\mathcal{V}$ is an {\em affine variety} in $\overline{K}^d$. Then $\dim(\mathcal{V})$ denotes the {\em dimension} of $\mathcal{V}$. Please see the precise definition of dimension of affine variety in \cite{cox}.

\begin{definition}\label{DEsus}
      Let ${\mathbb T}$ be a parametric regular decomposition of a given generic zero-dimensional system $\P$ in $K[U][X]$. Suppose $\mathcal{V}$ is an affine variety in $\overline {K}^{d}$ with $\dim(\mathcal{V})<d$.  If for any $a\in \overline{K}^d\backslash \mathcal{V}$,  ${\V}({\P}(a))=\cup_{{\T}\in \mathbb{T}}{\V}({\T}(a)\backslash \I{{\T}(a)})$ and ${\T}$  specializes well at $a$ for any $\T\in \mathbb{T}$, 
      then ${\mathbb T}$ is said to be a {\em generic regular decomposition} of $\P$  and $\mathcal{V}$ is said to be a {\em regular-decomposition-unstable (RDU) variety}  of $\P$ {\it w.r.t.} ${\mathbb T}$.
\end{definition}

  For any ${\bf P}\subset K[U][X]$,
$\V_{\overline{K}}(\bf P)$ denotes the set
$\{(a_1,  \ldots,  a_{d+n})\in \overline{K}^{d+n}|P(a_1,  \ldots,  a_{d+n})=0,  \forall P\in \P\}$. For any ${\bf B}\subset K[U]$,  ${\rm V}^{U}({\bf B})$ denotes the set
$\{(a_1,  \ldots,  a_d)\in \overline{K}^{d}|B(a_1,  \ldots,  a_d)=0,  \forall B\in \B\}.$
        For any $F\in K[U][X]$,
  the coefficients $B_1,  \ldots,  B_t$ of $F$  in $X$ are polynomials in $K[U]$.  Then
  ${\rm V}^{U}(F)$ denotes ${\rm V}^{U}(\{B_1,  \ldots,  B_t\})$.
   Note that for two finite subsets ${\bf P}$ and ${\bf H}$ of $K[U][X]$,
        ${\rm V}_{\overline{K(U)}}({\bf P}\backslash {\bf H})$ denotes the set
        ${\rm V}_{\overline{K(U)}}({\bf P})\backslash {\rm V}_{\overline{K(U)}}({\bf H})$. Similarly,  we can have $\V({\bf P}\backslash {\bf H})$,  $\V_{\overline{K}}({\bf P}\backslash {\bf H})$ and ${\rm V}^{U}({\bf P}\backslash {\bf H})$. The following Lemma \ref{UpdateSwell} is proposed in \cite{changbo}. Remark that the definition of regular chain in $K[U][X]$ in this paper is not exactly the same as that in \cite{changbo}  as mentioned in Remark \ref{REidea}. Therefore,   Lemma \ref{UpdateSwell} here is stated in our way.

\begin{lemma}\label{UpdateSwell}\mycite{changbo}
      Let ${\bf T}$ be a regular chain in $K[U][X]$. Then
      ${\bf T}$ specializes well at $a$ if and only if $a\in \overline{K}^d\backslash \V^{U}(\res{\I{{\T}}, {\T}})$.
\end{lemma}
\section{Theory and Algorithm}\label{sectionus}

\subsection{Weakly Relatively Simplicial Decomposition}\label{subsectionrsd}

In this section,  we introduce weakly relatively simplicial decomposition (WRSD)  in zero-dimensional case,  which is a weaker concept compared to relatively simplicial decomposition (RSD) proposed in \cite{zjzi}.
\begin{definition}\label{Dewrsd}
Let $\T$ be a zero-dimensional regular chain in $K[U][X]$ and $P\in K[U][X]$. Suppose  $\mathbb{H}$ and $\mathbb{G}$ are two finite sets of zero-dimensional regular chains in $K[U][X]$.  If

$(1)$ $\V_{\overline{K(U)}}({\T}\cup \{P\})=\cup_{{\bf H}\in {\mathbb H}}{\V}_{\overline{K(U)}}({\bf H})$ and

$(2)$ $\V_{\overline{K(U)}}({\T}\backslash P)=\cup_{{\bf G}\in {\mathbb G}}{\V}_{\overline{K(U)}}({\bf G})$,\\
then $(\mathbb{H}, \mathbb{G})$ is said to be a {\em WRSD} of $\T$ {\it w.r.t.} $P$ in $K[U][X]$.
\end{definition}
\begin{definition}\label{Deprsd}
Suppose $(\mathbb{H}, \mathbb{G})$ is a WRSD of a zero-dimensional regular $\T$ {\it w.r.t.} a polynomial $P$ in $K[U][X]$. 
The WRSD $(\mathbb{H}, \mathbb{G})$ is said to be {\em stable} at $a\in \overline {K}^{d}$  if

$(1)$ $\T$ specializes well at $a$,

$(2)$ $\V({\T}(a)\cup \{P(a)\})=\cup_{{\bf H}\in {\mathbb H}}\V({\bf H}(a))$ and  ${\bf H}$ specializes well at $a$ for any ${\bf H}\in \mathbb{H}$, and

$(3)$ $\V({\T}(a)\backslash P(a))=\cup_{{\bf G}\in \mathbb{G}}\V({\bf G}(a))$ and  ${\bf G}$ specializes well at $a$ for any ${\bf G}\in \mathbb{G}$.
\end{definition}
\begin{remark}\label{RErsd}
A stronger concept,  RSD,   was firstly introduced by Yang and Zhang in \cite{zjzi, zjziii} and the algorithm can be seen in \cite{xia, zjziiii}. Note that an RSD is a WRSD but the converse is not true. For instance,  $(\{\{x_1^2, x_2\}\}, \{\{x_1+u, x_2\}\})$ is a WRSD but not an RSD of $\{(x_1+u)x_1^2, x_2\}$ {\it w.r.t.} $x_1+x_2$ in $\mathbb{R}[u][x_1,x_2]$ because $\prem{x_1+x_2, \{x_1^2, x_2\}}=x_1\neq 0$.
\end{remark}

\begin{algorithm}\label{ALPRSD}
\DontPrintSemicolon
\SetAlgoCaptionSeparator{. }
 \caption{\APRSD}
    \KwIn{A zero-dimensional regular chain ${\bf T}=\{T_1, \ldots, T_n\}$ in $K[U][X]$,   a polynomial $P\in K[U][X]$,  variables $X=\{x_1, \ldots, x_n\}$ }
    \KwOut{[${\mathbb H}$, ${\mathbb G}$, $F$],  where
       $(\mathbb{H}, \mathbb{G})$ is a WRSD of $\T$ {\it w.r.t.} $P$ in $K[U][X]$ and
       $F$ is a polynomial in $K[U]$ such that for any $a\in \overline{K}^d\backslash {\V}^{U}(F)$, the WRSD $(\mathbb{H}, \mathbb{G})$ of $\T$ {\it w.r.t.} $P$ is stable at $a$.}
  $\mathbb{H}$:=$\emptyset$, $\mathbb{G}$:=$\emptyset$, $F$:=$\res{\I{\T}, \T}$\;
 \If{$P$ is not reduced {\it w.r.t.} $\T$}{return ${\APRSD}({\T}, \prem{P, {\T}}, X)$}
 \If{$P=0$}{return [$\{{\T}\}$, $\emptyset$, $F$]}
 \If{$\cls{P}=0$}{return [$\emptyset$, $\{{\T}\}$,  $P\cdot F$]}
 \If{$\cls{P}\neq n$}{$W$:=${\APRSD}(\{T_1, \ldots, T_{\cls{P}}\}, P, \{x_1, \ldots, x_{\cls{P}}\})$\;
$\mathbb{H}$:=${\tt map}(t\rightarrow t\cup \{T_{\cls{P}+1}, \ldots, T_n\}, W_1)$\;
$\mathbb{G}$:=${\tt map}(t\rightarrow t\cup \{T_{\cls{P}+1}, \ldots, T_n\}, W_2)$\;
return [$\mathbb{H}$, $\mathbb{G}$, $F\cdot W_3$]\;
}
\eIf{$\res{P, {\T}}\neq 0$}
   {$F$:=$F\cdot\res{P, \T}$, $\mathbb{G}$:=$\{\T\}$}
   {compute the regular subresultant chain $S_{d_\upsilon}, \ldots,  S_{d_1}, S_{d_0}$ of $T_n$ and $P$ {\it w.r.t.} $x_n$\;
   \eIf{$n=1$}{
    $\mathbb{H}$:=$\{\{S_{d_1}\}\}$\;
    $Q$:=$\pquo{T_1, S_{d_1}, x_1}$\;
    $\mathbb{G}$:=${\APRSD}(\{Q\}, P, X)_2$, $F$:=${\APRSD}(\{Q\}, P, X)_3$\;}
     {$X_{n-1}$:=$X\backslash \{x_n\}$\;
$\mathbb{H}_0$:=${\APRSD}(\{{T}_1, \ldots, {T}_{n-1}\}, S_{d_0}, X_{n-1})_1$\;
$\mathbb{G}_0$:=${\APRSD}(\{{T}_1, \ldots, {T}_{n-1}\}, S_{d_0}, X_{n-1})_2$\;
$\mathbb{G}$:=$\mathbb{G}\cup {\tt map}(t\rightarrow t\cup \{T_n\},  \mathbb{G}_0)$\;
$F$:=$F\cdot {\APRSD}(\{{T}_1, \ldots, {T}_{n-1}\}, S_{d_0}, X_{n-1})_3$\;
$i$:=$0$, $S_{d_{\upsilon+1}}$:=$T_n$\;
 \While{$\mathbb{H}_i\neq \emptyset$}
      {$i$:=$i+1$, $\mathbb{H}_i$:=$\emptyset$, $\mathbb{G}_i$:=$\emptyset$\;
      Let $R_{d_i}$ be the $d_i$-th principal subresultant coefficient of $T_n$ and $P$ {\it w.r.t.} $x_n$\;
      \For{${\bf H}\in \mathbb{H}_{i-1}$}{$\mathbb{H}_i$:=$\mathbb{H}_i\cup {\APRSD}({\bf H}, {R_{d_i}}, X_{n-1})_1$\;
                               $\mathbb{G}_i$:=$\mathbb{G}_i\cup {\APRSD}({\bf H}, {R_{d_i}},  X_{n-1})_2$\;
                               $F$:=$F\cdot {\APRSD}({\bf H}, {R_{d_i}}, X_{n-1})_3$}
       \For{${\bf G}\in \mathbb{G}_i$}{$\mathbb{H}$:=$\mathbb{H}\cup \{{\bf G}\cup \{S_{d_i}\}\}\}$\;
 $Q$:=$\pquo{T_n, S_{d_i}, x_n}$\;
      \If{$\deg(Q, x_n)>0$}{$\mathbb{G}$:=$\mathbb{G}\cup {\APRSD}({\bf G}\cup \{Q\}, P, X)_2$\;$F$:=$F\cdot {\APRSD}({\bf G}\cup \{Q\}, P, X)_3$}}
      }
}}
   return [$\mathbb{H}, \mathbb{G}, F$]\;
\end{algorithm}
Now we present Algorithm \ref{ALPRSD} for computing WRSDs\footnotemark, \footnotetext{Lines 2 and 3 of Algorithm \ref{ALPRSD} can be removed without loss of correctness.}
which is different from Algorithm {\tt RSD} proposed in \cite{zjzi}.  
Assume that ${\tt Alg}$ is a name of an algorithm and $p_1, \ldots, p_t$ is a sequence of inputs of this algorithm. If the output of ${\tt Alg}(p_1, \ldots, p_t)$ is a finite list [$q_1, \ldots, q_s$],  $q_i$ is denoted by ${\tt Alg}(p_1, \ldots, p_t)_i$ for any $i$ $(1\leq i\leq s)$ and also said to be the $i$th output of ${\tt Alg}(p_1, \ldots, p_t)$. Given a finite set $S=\{s_1, \ldots, s_t\}$ and a map $\phi$ on $S$,  ${\tt op}(S)$ denotes the finite sequence $s_1, \ldots, s_t$  and ${\tt map}(s\rightarrow \phi(s),  S)$ denotes the set $\phi(S)$.

Before showing the termination and the correctness of Algorithm \ref{ALPRSD},  we need to prepare some statements. In the following discussion,  we assume that the readers are familiar with the theories of subresultants. The precise definitions of subresultant chain and regular subresultant chain can be seen in \cite{mishra, kahoui}  and Lemma \ref{LEhuantongtai} can be found in \cite{kahoui, xia}.
 \begin{lemma}\label{LEhuantongtai}\mycite{kahoui, xia}
         Let $\phi:\emph{R}\rightarrow \widetilde{\emph
         R}$ be a ring homomorphism.  Denote also by $\phi$ the induced homomorphism $\tilde{\phi}:\emph{R}[x]\rightarrow \widetilde{\emph R}[x]$,   where both $\emph{R}$ and $\widetilde{\emph R}$ are integral domains.
     Suppose $F$ and $G$ are polynomials in
     $\emph{R}[x]$ and $b$ and $c$ are the leading coefficients of $F$ and $G$  respectively. Assume that
     $m=\deg(F, x)\geq l=\deg(G, x)>0$ and $\widetilde{m}=\deg(\phi(F), x)\geq \widetilde{l}=\deg(\phi(G), x)>0$.
     If $\widetilde{m}>\widetilde{l}$,   let $\widetilde{\mu}=\widetilde{m}-1$,
     otherwise,  $\widetilde{\mu}=\widetilde{m}$. Suppose $S_j$ is the $j$-th subresultant of $F$ and $G$ {\it w.r.t.} $x$ and $\widetilde{S_j}$ is the $j$-th subresultant of $\phi(F)$ and $\phi(G)$ {\it w.r.t.} $x$. Then $\phi(S_j)=\delta\cdot
     \widetilde{S_j}$ for any $j$ $(0\leq j<\widetilde{\mu})$,   where\par
     \[\delta=
     \left\{
       \begin{array}{ll}
         1,   & \hbox{$\phi(b)\cdot\phi(c)\neq 0$,  } \\
         \phi(b)^{l-\widetilde{l}},   & \hbox{$\phi(b)\neq
         0$ and $\phi(c)=0$,  } \\
         (-1)^{(m-\widetilde{m})(l-j)} \phi(c)^{m-\widetilde{m}},   &
         \hbox{$\phi(b)=0$ and $\phi(c)\neq 0$,  } \\
         0,   & \hbox{$\phi(b)=\phi(c)=0$. }
       \end{array}
     \right.\]
 Furthermore,  if $m>l$,  let $\mu=m-1$,  otherwise,  let $\mu=m$. Suppose $R_j$ is the $j$-th principal subresultant coefficient of $F$ and $G$ {\it w.r.t.} $x$. Then $\phi(G)=0$ if $\phi(b)\neq 0$ and $\phi(R_j)=0$ for any $j$ $(0\leq j\leq \mu)$.
\end{lemma}
Roughly speaking,  Algorithm \ref{ALPRSD} is based on Lemma \ref{LEzerode},  which is inspired by the analogous results presented in \cite{wangi, wang}. Note that the results shown in Lemma \ref{LEzerode} is not covered by that in \cite{wangi, wang}.
\begin{lemma}\label{LEzerode}
Given two polynomials $F$ and $G$ in $K[U][X]$ $(0<\deg(G, x_n)< \deg(F, x_n))$,  suppose $S_{d_\upsilon}, \ldots, S_{d_1}, S_{d_0}$ is the regular subresultant chain of $F$ and $G$ {\it w.r.t.} $x_n$. Let $S_{d_{\upsilon+1}}=F$. Assume that  $R_{d_i}$ is the $d_i$th principal subresultant coefficient of $F$ and $G$ {\it w.r.t.} $x_n$ for any $i$ $(0\leq i\leq \upsilon+1)$ and $Q_{d_i}$ is the pseudoquotient of $F$ and $S_{d_i}$ {\it w.r.t.} $x_n$ for any $i$ $(1\leq i\leq \upsilon)$.  Then

$(1)$${\V}_{\overline{K(U)}}(\{F, G\}\backslash \I{F})=\cup_{i=1}^{\upsilon+1}{\V}_{\overline{K(U)}}(\{S_{d_i},  R_{d_{i-1}}, \ldots,  R_{d_0}\}\backslash \I{F} R_{d_i})$;

$(2)$${\V}_{\overline{K(U)}}(F\backslash G\I{F})={\V}_{\overline{K(U)}}(F\backslash G\I{F} R_{d_0})\cup \cup_{i=1}^{\upsilon}{\V}_{\overline{K(U)}}(\{Q_{d_i},  R_{d_{i-1}}, \ldots, R_{d_0}\}\backslash G\I{F} R_{d_i})$.
\end{lemma}
\begin{proof}
Assume that $S_{\mu+1}, S_{\mu}, \ldots, S_1, S_0$ is the subresultant chain of $F$ and $G$
{\it w.r.t.} $x_n$. Remark that $S_{d_0}=S_0=\res{F, G}$, $S_{\mu}=G$, $S_{\mu+1}=F$ and $S_{d_\upsilon}=\I{G}^cG$ where $c$ is a non-negative integer.

(1)Assume that ${\V}_{\overline{K(U)}}(\{F, G\}\backslash \I{F})\neq \emptyset$. For any $(a_1, \ldots, a_n)\in {\V}_{\overline{K(U)}}(\{F, G\}\backslash\I{F})$,  let $b=(a_1, \ldots, a_{n-1})$. If $G(b)=0$,  by the definition of principal subresultant coefficient,  $R_{d_i}(b)=0$ and thus $R_{d_i}(a)=0$ for any $i$ $(1\leq i\leq \upsilon)$. Hence,  $(a_1, \ldots, a_n)\in {\V}_{\overline{K(U)}}(\{S_{d_{\upsilon+1}},  R_{d_\upsilon}, \ldots, R_{d_0}\}\backslash \I{F}R_{d_{\upsilon+1}})$.
If $\deg(G(b), x_n)>0$,  since $\deg(G(b), x_n)<\deg(F(b), x_n)=\deg(F, x_n)$,  it is reasonable to assume that the subresultant chain of $F(b)$ and $G(b)$ {\it w.r.t.} $x_n$ is $\widetilde{S}_{\mu+1}, \widetilde{S}_{\mu}, \ldots, \widetilde{S}_1, \widetilde{S}_0$ and  the associated principal coefficients are $ \widetilde{R}_{\mu+1}, \widetilde{R}_{\mu}, \ldots, \widetilde{R}_1, \widetilde{R}_0$. Note that $\widetilde{S}_{\mu}=G(b), \widetilde{S}_{\mu+1}=F(b)$ and by Lemma \ref{LEhuantongtai},  we know that $S_j(b)=\I{F}(b)^{r_j}\widetilde{S}_j$  where $r_j$ is a non-negative integer for any $j$ $(1\leq j\leq \mu+1)$. According to the theories of subresultant chains,  there exists an integer $j$ $(1\leq j\leq \mu)$ such that ${\widetilde{R}_j}\neq 0$ and ${\widetilde{R}_0}=\ldots={\widetilde{R}_{j-1}}=0$. Then $R_j(b)\neq 0$ and $R_0(b)=\ldots=R_{j-1}(b)=0$. In addition,  $\widetilde{S}_j$ is the greatest common divisor of $F(b)$ and $G(b)$ in $\overline{K(U)}[x_n]$ and $\deg(\widetilde{S}_j, x_n)=j$.   Hence $\widetilde{S}_j(a_n)=0$ by $F(b)(a_n)=G(b)(a_n)=0$.  Note that $\deg(S_j, x_n)=\deg(S_j(b),  x_n)=\deg(\widetilde{S}_j, x_n)=j$,  so there exists some $i$ $(1\leq i\leq \upsilon)$ such that $d_i=j$. Therefore,  $(a_1, \ldots, a_n)\in {\V}_{\overline{K(U)}}(\{S_{d_i}, {R_{d_{i-1}}}, \ldots,  {R_{d_0}}\}\backslash \I{F}{R_{d_i}})$.

On the other hand,  for any $(a_1, \ldots, a_n)\in {\V}_{\overline{K(U)}}(\{S_{d_{\upsilon+1}},  {R_{d_\upsilon}}, \ldots, {R_{d_0}}\}\backslash \I{F}{R_{d_{\upsilon+1}}})$,  let $b=(a_1, \ldots, a_{n-1})$. As ${R_{d_i}}(b)=0$ for any $i$ $(1\leq i\leq \upsilon)$,  $G(b)=0$ follows from Lemma \ref{LEhuantongtai}. Hence,  $(a_1, \ldots, a_n)\in {\V}_{\overline{K(U)}}(\{F, G\}\backslash \I{F})$.  For any $i$ $(1\leq i\leq \upsilon)$ and for any  $(a_1, \ldots, a_n)\in {\V}_{\overline{K(U)}}(\{S_{d_i}, {R_{d_{i-1}}}, \ldots, {R_{d_0}}\}\backslash \I{F}{R_{d_i}})$,  it is not difficult to check $(a_1, \ldots, a_n)\in {\V}_{\overline{K(U)}}(\{F, G\}\backslash \I{F})$ similarly as what has been discussed in the last paragraph.

(2) The proof is similar to that of (1).
\end{proof}
\begin{remark}
If $F$ and $G$ are polynomials in $\overline{K}[X]$ $(0<\deg(G, x_n)< \deg(F, x_n))$,  suppose $S_{d_\upsilon}, \ldots, S_{d_1}, S_{d_0}$ is the regular subresultant chain of $F$ and $G$ {\it w.r.t.} $x_n$. Let $S_{d_{\upsilon+1}}=F$. Assume that  $R_{d_i}$ $(0\leq i\leq \upsilon+1)$ is the $d_i$th principal subresultant coefficient of $F$ and $G$ {\it w.r.t.} $x_n$ and $Q_{d_i}$ $(1\leq i\leq \upsilon)$ is the pseudoquotient of $F$ and $S_{d_i}$ {\it w.r.t.} $x_n$.   Similarly,  we have\\
$(1)$${\V}(\{F, G\}\backslash \I{F})=\cup_{i=1}^{\upsilon+1}{\V}(\{S_{d_i},  R_{d_{i-1}}, \ldots,  R_{d_0}\}\backslash \I{F} R_{d_i})$;\\
$(2)$${\V}(F\backslash G\I{F})={\V}(F\backslash G\I{F} R_{d_0})\cup \cup_{i=1}^{\upsilon}{\V}(\{Q_{d_i},  R_{d_{i-1}}, \ldots, R_{d_0}\}\backslash G\I{F} R_{d_i})$.
\end{remark}
\begin{lemma}\label{LESPRES}
Let $P\in K[U][X]$ and ${\bf T}=\{T_1,  \ldots,  T_n\}$ be a zero-dimensional regular chain in $K[U][X]$.  If  $S_0=\res{P,  {\bf T}}\neq 0$,  then $\res{P(a), {\bf T}(a)}\neq 0$ for any $a\in\overline{K}^d\backslash{\V}^U(S_0\res{\I{\T}, {\T}})$.
\end{lemma}
\begin{proof}
 It is not difficult to prove the conclusion by induction on $n$.
\end{proof}

\begin{lemma}\label{LERED}
Given a zero-dimensional regular chain ${\T}=\{T_1,  \ldots,  T_n\}$ in $K[U][X]$ and a polynomial $P\in K[U][X]$,  suppose $P_1=\prem{P, {\T}}$. Then ${\V}_{\overline{K(U)}}({\T}\cup \{P\})={\V}_{\overline{K(U)}}({\T}\cup \{P_1\})$ and ${\V}_{\overline{K(U)}}({\T}\backslash P)={\V}_{\overline{K(U)}}({\T}\backslash P_1)$. Furthermore,  ${\V}({\T}(a)\cup \{P(a)\})={\V}({\T}(a)\cup \{P_1(a)\})$ and ${\V}({\T}(a)\backslash P(a))={\V}({\T}(a)\backslash P_1(a))$ for any $a\in \overline{K}^d\backslash {\V}^U(\res{\I{\T}, {\T}})$.
\end{lemma}
\begin{proof}
It is easy to prove the conclustion by the definition of successive pseudodivision and Lemma \ref{UpdateSwell}.
\end{proof}
\begin{theorem}\label{THwrsd}
Algorithm \ref{ALPRSD} terminates correctly.
\end{theorem}
\begin{proof}
The termination is similar as the termination of Algorithm RSD in \cite{zjzi,zjziii}. For a given zero-dimensional regular chain ${\T}=\{T_1, \ldots, T_n\}$ in $K[U][X]$ and a polynomial $P$ in $K[U][X]$,  let ${\APRSD}({\T}, P, X)=[\mathbb{H},  \mathbb{G}, F]$. Now we prove the correctness by induction on the recursive depth $h$ of ${\APRSD}({\T}, P, X)$. Note that we only need to prove that $\mathbb{H},  \mathbb{G}$ and $F$ satisfy the conditions stated in Definitions \ref{Dewrsd} and \ref{Deprsd}.

 When $h=1$,  the conclusion follows from Lemma \ref{UpdateSwell} and Lemma \ref{LESPRES}. Assume that the conclusion holds when $h<N$ $(N>1)$. Suppose $h=N$. Then ${\APRSD}({\T}, P, X)$ can return at Line 3, Line 12, or Line 41.   If ${\APRSD}({\T}, P, X)$ returns at Line 3, the conclusion follows from the induction hypothesis and Lemma \ref{LERED}. If ${\APRSD}({\T}, P, X)$ returns at Line 12,  the conclusion follows from the induction hypothesis. Now we prove the conclusion when ${\APRSD}({\T}, P, X)$ returns at Line 41, which means $P$ is reduced {\it w.r.t.} $\T$,  $\mvar{P}=x_n$ and $\res{P, {\T}}=0$. Suppose $S_{\mu+1},  S_\mu, \ldots, S_1, S_0$ is the subresultant chain of $T_n$ and $P$ {\it w.r.t.} $x_n$ in $K[U, X_{n-1}][x_n]$ where $X_{n-1}=X\backslash \{x_n\}$ and $S_{d_\upsilon}, \ldots,  S_{d_1}, S_{d_0}$ is the associated regular subresultant chain. Note that $\deg(T_n, x_n)>\deg(P, x_n)>0$ since $P$ is reduced {\it w.r.t.} $\T$.

  If $n=1$,  $S_{d_1}$ is the greatest common divisor of $T_1$ and $P$ in $K[U][x_1]$ and hence ${\V}_{\overline{K(U)}}(\{T_1, P\})\\={\V}_{\overline{K(U)}}(\{S_{d_1}\})$.  Then condition (1)  in Definition \ref{Dewrsd} holds. Suppose $Q=\pquo{T_1,S_{d_1},x_1}$. Remark that $\deg(Q, x_1)>0$ and there exists a positive integer $k$  such that $k\geq 2$  and ${\rm I}_{S_{d_1}}^kT_1=S_{d_1}Q$. Thus ${\V}_{\overline{K(U)}}({\T})={\V}_{\overline{K(U)}}(S_{d_1}Q)$. Note that ${\V}_{\overline{K(U)}}(S_{d_1})\subset {\V}_{\overline{K(U)}}(P)$. So ${\V}_{\overline{K(U)}}({\T}\backslash P)={\V}_{\overline{K(U)}}(S_{d_1}Q\backslash P)={\V}_{\overline{K(U)}}(Q\backslash P)$. Therefore  condition (2) in Definition \ref{Dewrsd} follows from the induction hypothesis.
  Remark that $\I{S_{d_1}}$ is a factor of $\I{Q}$ and according to Algorithm \ref{ALPRSD},  $\I{Q}$ is a factor of $F$. Thus
   for any $a\in \overline{K}^d\backslash {\V}^{U}(F)$, $\I{T_1}(a)\neq 0$  and $\deg(P(a), x_1)\geq d_1>0$  by the definition of subresultant. Obviously,  condition (1) in Definition \ref{Deprsd} holds. Besides, according to Lemma \ref{LEhuantongtai},  $S_{d_1}(a)$ is the great common divisor of $P(a)$ and $T_1(a)$ in $\overline{K}[x_1]$. Thus ${\V}({\T}(a)\cup P(a))={\V}(S_{d_1}(a))$ and  condition (2) in Definition \ref{Deprsd} holds.  Since  ${{\rm I}_{S_{d_1}}(a)}^kT_1(a)=S_{d_1}(a)Q(a)$ and ${\V}(S_{d_1}(a))\subset {\V}(P(a))$,  ${\V}({\T}(a)\backslash P(a))={\V}(S_{d_1}(a)Q(a)\backslash P(a))={\V}(Q(a)\backslash P(a))$. Therefore  condition (3) in Definition \ref{Deprsd} follows from the induction hypothesis.

 If $n>1$,  let $S_{d_{\upsilon+1}}=S_{\mu+1}$ and  ${\T}_{n-1}=\{T_1, \ldots, T_{n-1}\}$.
 Suppose $R_{d_i}$ is the principal subresultant  coefficient of $T_n$ and $P$ {\it w.r.t.} $x_n$ for any $i$ ($0\leq i\leq \upsilon+1$) and assume that $\mathbb{H}_0={\APRSD}({\T}_{n-1}, S_{d_0}, X_{n-1})_1$ and $\mathbb{G}_0={\APRSD}({\T}_{n-1}, S_{d_0}, X_{n-1})_2$. Remark that $\mathbb{H}_0\neq \emptyset$ because $\res{R_{d_0}, {\T}_{n-1}}=\res{P, {\T}}=0$.  For any $i$ $(1\leq i)$,  let $\mathbb{H}_i=\cup_{{\bf H}\in \mathbb{H}_{i-1}}{\APRSD}({\bf H}, R_{d_i}, X_{n-1})_1$ and $\mathbb{G}_i=\cup_{{\bf H}\in \mathbb{H}_{i-1}}{\APRSD}({\bf H}, R_{d_i}, X_{n-1})_2$ until  there exists an integer $l$ $(1\leq l\leq \upsilon+1)$ such that $\mathbb{H}_l=\emptyset$. That means $\mathbb{H}_l=\emptyset$ and $\mathbb{H}_j\neq \emptyset$ for any $j$ $(0\leq j<l)$. We can always get this integer $l$  owing to the fact that $S_{d_{\upsilon+1}}=T_n$.   Then we have two sequences $\mathbb{H}_0,\mathbb{H}_1,\ldots,\mathbb{H}_l$ and   $\mathbb{G}_0,\mathbb{G}_1,\ldots,\mathbb{G}_l$.
 Let $L_1=\{i|1\leq i\leq l, \mathbb{G}_i\neq \emptyset\}$.
 According to Algorithm \ref{ALPRSD}, the first output of ${\APRSD}({\T},P,X)$ is
 $\mathbb{H}=\cup_{i\in L_1}\cup_{{\bf G}\in {\mathbb G}_i}({\bf G}\cup \{S_{d_i}\})$. It is not difficult to see that $\mathbb{H}$ is a finite set of zero-dimensional regular chains in $K[U][X]$.
 By Lemma \ref{LEzerode}(1), we know that
${\V}_{\overline{K(U)}}({\T}\cup \{P\})
=\cup_{i=1}^{\upsilon+1}({\V}_{\overline{K(U)}}({\T}_{n-1})\cap {\V}_{\overline{K(U)}} (\{S_{d_i},R_{d_{i-1}},\ldots,R_{d_0}\}\backslash \I{T_n}R_{d_i}))$.
For any $i$ $(1\leq i\leq \upsilon+1)$, If $i\in L_1$, according to the induction hypothesis and the construction of ${\mathbb G}_i$, we get
${\V}_{\overline{K(U)}}({\T}_{n-1})\cap {\V}_{\overline{K(U)}} (\{S_{d_i},R_{d_{i-1}},\ldots,R_{d_0}\}\backslash \I{T_n}R_{d_i})=\cup_{{\bf G}\in {\mathbb G}_{i}}{\V}_{\overline{K(U)}}({\bf G}\cup \{S_{d_i}\})$.
If $l<i\leq \upsilon+1$, according to the induction hypothesis and ${\mathbb H}_{l}=\emptyset$, similarly, we know that
${\V}_{\overline{K(U)}}({\T}_{n-1})\cap {\V}_{\overline{K(U)}} (\{S_{d_i},R_{d_{i-1}},\ldots,R_{d_0}\}\backslash \I{T_n}R_{d_i})=\emptyset$.
If $1\leq i\leq l$ and $i\not\in L_1$, similarly, we get$
{\V}_{\overline{K(U)}}({\T}_{n-1})\cap {\V}_{\overline{K(U)}} (\{S_{d_i},R_{d_{i-1}},\ldots,R_{d_0}\}\backslash \I{T_n}R_{d_i})=\emptyset$.
Therefore, ${\V}_{\overline{K(U)}}({\T}\cup \{P\})=\cup_{i\in L_1}\cup_{{\bf G}\in {\mathbb G}_i}{\V}_{\overline{K(U)}}({\bf G}\cup \{S_{d_i}\})=\cup_{{\bf H}\in \mathbb{H}}{\V}_{\overline{K(U)}}({\bf H})$ and hence  condition (1) in Definition \ref{Dewrsd} holds. Furthermore, as discussed above, we figure out that ${\V}_{\overline{K(U)}}({\T}\cup \{P\})=\emptyset$ if and only if $\mathbb{H}=\emptyset$. Actually, when $\res{P,{\T}}=0$, ${\V}_{\overline{K(U)}}({\T}\cup \{P\})$ cannot be $\emptyset$ according to Proposition \ref{PRRCP2}(3) and thus $\mathbb{H}\neq \emptyset$.
Similarly, we can prove that  condition (2) in Definition \ref{Dewrsd} holds on the basis of Lemma \ref{LEzerode}(2). Besides, it also can be shown that ${\V}_{\overline{K(U)}}({\T}\backslash P)=\emptyset$ if and only if $\mathbb{G}=\emptyset$.

For any $a\in \overline{K}^d\backslash {\V}^U(F)$, $\T$ specializes well at $a$ by Line 1 and Lemma \ref{UpdateSwell} and thus  condition (1) in Definition \ref{Deprsd} holds. It is also easy to check that ${\bf H}$ specializes well at $a$ for any ${\bf H}\in \mathbb{H}$ by the induction hypothesis
 and we only need to prove that ${\V}({\T}(a)\cup \{P(a)\})=\cup_{{\bf H}\in \mathbb{H}}{\V}({\bf H}(a))$. If $\deg(P(a),x_n)=0$,  it is easy to see that ${\V}({\T}(a)\cup \{P(a)\})=\cup_{{\bf H}\in \mathbb{H}}{\V}({\bf H}(a))$.
 If $\deg(P(a),x_n)>0$,  it is reasonable to assume that the subresultant chain of $T_n(a)$ and $P(a)$ {\it w.r.t.} $x_n$ is $\widetilde{S}_{\mu+1},\widetilde{S}_{\mu},\ldots,\widetilde{S}_0$. By Lemma \ref{LEhuantongtai}, we know that $S_i(a)={\I{T_n}(a)}^{r_i}\widetilde{S_i}$ where $r_i$ is a non-negative integer for any $i$ $(0\leq i\leq \mu+1)$. Suppose $L_2=\{i|1\leq i\leq \upsilon+1, R_{d_i}(a)\neq 0\}$. It is not difficulty to check that $L_1\subset L_2$ by the induction hypothesis and it is reasonable to assume that $L_2=\{j_1,\ldots,j_k,j_{k+1}\}$ $(k\geq 1)$ such that $0<d_{j_1}<\ldots<d_{j_k}<d_{j_{k+1}}=d_{\upsilon+1}$. Then $\widetilde{S}_{0},\widetilde{S}_{d_{j_1}},\ldots,\widetilde{S}_{d_{j_k}}$ is  the regular subresultant chain of $T_n(a)$ and $P(a)$ {\it w.r.t.} $x_n$.  By Remark \ref{LEzerode}(1), ${\V}({\T}(a)\cup \{P(a)\})=\cup_{j_t\in L_2}({\V}({\T}_{n-1}(a))\cap {\V}(\{\widetilde{S}_{d_{j_t}},\widetilde{R}_{d_{j_{t-1}}},\ldots,\widetilde{R}_{d_0}\}\backslash \I{T_n(a)}\widetilde{R}_{d_{j_t}}))$.
For any $j_t\in L_2$, if $j_t\in L_1$,
then by the induction hypothesis,
${\V}({\T}_{n-1}(a))\cap {\V}(\{\widetilde{S}_{d_{j_t}},\widetilde{R}_{d_{j_{t-1}}},\ldots,\widetilde{R}_{d_{j_1}},\widetilde{R}_{d_0}\}\backslash \I{T_n(a)}\widetilde{R}_{d_{j_t}})=\cup_{{\bf G}\in {\mathbb G}_{j_t}}{\V}({\bf G}(a)\cup \{S_{d_{j_t}}(a)\})$.
For any $j_t\in L_2\backslash L_1$, if $j_t\leq l$, by the induction hypothesis, $\cup_{{\bf H}\in {\mathbb H}_{{j_t}-1}}{\V}({\bf H}(a)\backslash R_{d_{j_t}}(a))=\emptyset$ since ${\mathbb G}_{j_t}=\emptyset$.
Then ${\V}({\T}_{n-1}(a))\cap {\V}(\{\widetilde{S}_{d_{j_t}},\widetilde{R}_{d_{j_{t-1}}},\ldots,\widetilde{R}_{d_0}\}\backslash \I{T_n(a)}\\\widetilde{R}_{d_{j_t}})
=\emptyset$.
If $j_t>l$, by the induction hypothesis, $\cup_{{\bf H}\in {\mathbb H}_{l-1}}{\V}({\bf H}(a)\cup  \{R_{d_l}(a)\})=\emptyset$ since ${\mathbb H}_{l}=\emptyset$.
Then ${\V}({\T}_{n-1}(a))\cap {\V}(\{\widetilde{S}_{d_{j_t}},\widetilde{R}_{d_{j_{t-1}}},\ldots,\widetilde{R}_{d_0}\}\backslash \I{T_n(a)}\widetilde{R}_{d_{j_t}})=\emptyset$.
Therefore, ${\V}({\T}(a)\cup \{P(a)\})=\cup_{j_t\in L_1}\cup_{{\bf G}\in \mathbb{G}_{j_t}}{\V}({\bf G}(a)\cup \{S_{j_t}(a)\})=\cup_{{\bf H}\in \mathbb{H}}{\V}({\bf H}(a))$ and hence  condition (2) in Definition \ref{Deprsd} holds. Similarly,  we can check that  condition (3) in Definition \ref{Deprsd} holds by Remark \ref{LEzerode}(2).
\end{proof}

\begin{remark}
 For a WRSD $(\mathbb{H},\mathbb{G})$ of a given zero-dimensional regular $\T$ {\it w.r.t.} a given polynomial $P$ in $K[U][X]$, consider the set $S=\{a\in \overline{K}^n|$ the WRSD $(\mathbb{H},\mathbb{G})$ of $\T$ {\it w.r.t.} $P$ is stable at $a\}$. Obviously, if $F$ is the third output of ${\APRSD}({\T}, P, X)$, then ${\V}^{U}(F)\subset S$. But we cannot prove that $S\subset {\V}^{U}(F)$ and we do not know how to compute $S$ or any set $S_1$ such that ${\V}^{U}(F)\subsetneq S_1\subset S$ efficiently.  However, it may demand huge amount extra computation to enlarge the set  ${\V}^{U}(F)$ slightly. It is interesting to develop algorithms for computing $S$ efficiently in the future.
\end{remark}

\subsection{Computing RDU Varieties}

We present the main result of this paper in this section.  Algorithm \ref{ALsus} shows how to compute a generic regular decomposition and the associated RDU variety  of a given generic zero-dimensional system\footnotemark\footnotetext{Whether a given system is generic zero-dimensional can be checked by  Algorithm \ref{ALsus} itself. Hence, we assume that the input system of Algorithm \ref{ALsus} is always generic zero-dimensional.} simultaneously, in which Algorithm \ref{ALZDToRC} plays a key role.
\begin{theorem}\label{THZDToRC}
Algorithm \ref{ALZDToRC} terminates correctly.
\end{theorem}
\begin{proof}
      If the input ${\bf T}$ is a regular chain,  then the termination holds obviously and the correctness follows from Lemma \ref{UpdateSwell}. Now we assume that $\T$ is not a regular chain and let $k$ be the minimal integer $k$ $(1\leq k<n)$ such that ${\T}_{k}=\{T_1, \ldots, T_k\}$ is a regular chain and $\{T_1, \ldots, T_{k+1}\}$ is not. Remark that  this assumption is reasonable owing to the fact that at least $\{T_1\}$ is a regular chain in $K[U][X]$. Let ${\T}_{n-k}=\{T_{k+1},\ldots,T_n\}$.
 Assume that ${\tt ZDToRC}({\T},X)$ does not terminate.  Then we can get at least one regular chain ${\bf R}\in {\APRSD}({\T}_{k},\I{T_{k+1}},\{x_1,\ldots,x_k\})_2$ such that ${\tt ZDToRC}({\bf R}\cup {\T}_{n-k},X)$ cannot terminate. According to Algorithm \ref{ALZDToRC}, there exists $k_2$ ($1\leq k_2<n$) such that $k_2$ is the minimal integer such that ${\bf R}\cup \{T_{k+1},\ldots,T_{k_2}\}$ is a regular chain but ${\bf R}\cup\{T_{k+1},\ldots,T_{k_2+1}\}$ is not. It is easy to check that $k_2>k$. Since ${\tt ZDToRC}({\bf R}\cup {\T}_{n-k},X)$ does not terminate, the rest can be done in the same manner and we can get an infinite sequence of positive integers $k=k_1<k_2<\ldots<k_t<\ldots$. Note that all positive integers in this infinite sequence must be no more than $n$ and it is impossible obviously. Therefore,  Algorithm \ref{ALZDToRC} terminates.  Then it is not difficult to prove the correctness by induction on the recursive depth $h$.
\end{proof}
\begin{algorithm}\label{ALZDToRC}
\SetAlgoCaptionSeparator{. }
\caption{\tt{ZDToRC}}
\DontPrintSemicolon
\KwIn{A triangular set ${\bf T}=\{T_1, \ldots, T_n\}$ in $K[U][X]$ satisfying $\mvar{{\T}}=X$, variables $X=\{x_1, \ldots, x_n\}$. }
\KwOut{[$\mathbb{G}$, $F$], where
             $\mathbb{G}$  is a  finite set  of zero-dimensional regular chains in $K[U][X]$ such that  ${\V}_{\overline{K(U)}}({\bf T}\backslash \I{{\T}})=\cup_{{\bf G}\in \mathbb{G}}{\V}_{\overline{K(U)}}({\bf G})$,
            and $F$ is a polynomial in $K[U]$ such that for any $a\in \overline{K}^{d}\backslash {\V}^{U}(F)$, ${\V}({\T}(a)\backslash \I{{\T}}(a))=\cup _{{\bf G}\in \mathbb{G}}{\V}({\bf G}(a))$ and ${\bf G}$ specializes well at $a$ for any ${\bf G}\in \mathbb{G}$ if $\mathbb{G}\neq \emptyset$.}
            \eIf{$\T$ is a regular chain}{return [$\{\T\}$, $\res{\I{{\T}}, {\T}}$]\; }{Find the minimal integer $k$ $(1\leq k<n)$ such that ${\T}_{k}=\{T_1, \ldots, T_{k}\}$ is a regular chain and  ${\T}_{k+1}=\{T_1, \ldots, T_{k},T_{k+1}\}$ is not a regular chain\; \label{ZDToRCyjminimal}}
            $W$: =$\tt{WRSD}$$({\T}_{k}, \I{{T}_{k+1}}, \{x_1,\ldots,x_k\})$\;
            \If{$W_2=\emptyset$}{return [$\emptyset, W_3$]}
            $F$: =$W_3$, $\mathbb{G}$: =$\emptyset$\;
            \For{${\bf T}$ in $W_2$}
         {${\bf R}$: =$\{\op{{\bf T}}, T_{k+1}, \ldots, T_n\}$\;
        $\mathbb{G}$: =$\mathbb{G}\cup $$\tt{ZDToRC}$$({\bf R},X)_1$,
       $F$: =$F\cdot $$\tt{ZDToRC}$$({\bf R},X)_2$\; }
      return [$\mathbb{G}$, $F$]\;
\end{algorithm}

\begin{lemma}\label{LESPZERO}
Suppose ${\V}_{\overline{K(U)}}({\T})\subset {\V}_{\overline{K(U)}}(\P)$ for a zero-dimensional regular chain $\T$ in $K[U][X]$ and a system ${\P}\subset K[U][X]$. Then ${\V}({\T}(a))\subset {\V}({\P}(a))$ for any $a\in \overline{K}^d\backslash {\V}^U(\res{\I{\T},{\T}})$.
\end{lemma}
\begin{proof}
By Proposition \ref{PRRCP2}(3), ${\V}_{\overline{K(U)}}({\T}\backslash \I{{\T}})={\V}_{\overline{K(U)}}({\T})$. Then ${\V}_{\overline{K(U)}}({\T}\backslash \I{{\T}})\subset {\V}_{\overline{K(U)}}({\P})$. By Proposition \ref{PRRCP2}(2), ${\P}\subset\sqrt{\SAT{{\T}}_{K[U][X]}}$ and hence for any $P\in \P$, there exists a positive integer $k$ such that $P^k\in \SAT{{\T}}_{K[U][X]}$. By Proposition \ref{PRRCP2}(1), $\prem{P^k, {\T}}=0$.  Remark that  $P^k\in K[U][X]$ and ${\T}\subset K[U][X]$, so ${\V}_{\overline{K}}({\T}\backslash \I{{\T}})\subset {\V}_{\overline{K}}(P^k)={\V}_{\overline{K}}(P)$. For any $a=(a_1,\ldots,a_d)\in \overline{K}^d\backslash {\V}^U(\res{\I{\T},{\T}})$, ${\T}(a)$ is a zero-dimensional regular chain in $\overline{K}[X]$ and $\I{\T}(a)\neq 0$ by Lemma \ref{UpdateSwell}. Hence $\res{\I{{\T}}(a),{\T}(a)}=\res{\I{{\T}(a)},{\T}(a)}\neq 0$ and thus for any $b=(b_1,\ldots,b_n)\in {\V}({\T}(a))$, $b\not\in {\V}(\I{{\T}}(a))$.  That implies $(a_1,\ldots,a_d,b_1,\\\ldots,b_n)\in {\V}_{\overline{K}}({\T}\backslash \I{{\T}})\subset {\V}_{\overline{K}}({\P})$. Therefore, $b\in \V({\P}(a))$ and then ${\V}({\T}(a))\subset \V({\P}(a))$.
\end{proof}
\begin{theorem}
Algorithm \ref{ALsus} terminates correctly.
\end{theorem}
\begin{proof}
Since the termination follows from the termination of Algorithm \ref{ALPRSD} and Algorithm \ref{ALZDToRC}, we only need to prove the correctness.  Assume that $\P$ is a generic zero-dimensional system in $K[U][X]$ and $\{{\bf C}_1, \ldots, {\bf C}_m\}$ is a Wu's decomposition of ${\bf P}\subset K[U][X]$ computed by Wu's method. According to Wu's method and Algorithm \ref{ALZDToRC}, we know that the claim (1) in the specification of Algorithm \ref{ALsus} holds.

 Now we prove the claim (2) in the specification of Algorithm \ref{ALsus} by induction on $m$. If $m=1$,  ${\bf C}_1$ is a characteristic set of $\P$ in $K[U][X]$ and ${\bf C}_1=\{C_1\}\subset K[U]$ by Wu's method. That means ${\V}_{\overline{K(U)}}({\P})=\emptyset$. According to Algorithm \ref{ALsus}, the first output of ${\ARDU}({\P},X)$ is
$\emptyset$ and  the second output is exactly $C_1$. In fact,
for any $a\not\in {\V}^{U}(C_1)$, $C_1(a)\in \overline{K}$. Note that $C_1(a)\in \ideal{{\P}(a)}_{\overline{K}[X]}$ by $C_1\in \ideal{{\P}}_{K[U][X]}$. Thus ${\V}({\P}(a))\subset {\V}(C_1(a))=\emptyset$ and the claim (2) in the specification of Algorithm \ref{ALsus} holds.

Assume that the conclusion holds for $m<N$ $(N>1)$.  If $m=N$, suppose ${\bf C}_1=\{{C_1}_1,\ldots,{C_1}_t\}$ is the characteristic set of $\P$ computed by Wu's method. Since $m>1$, we know that ${\bf C}_1\not\subset K[U]$ and ${\V}_{\overline{K(U)}}({\P})={\V}_{\overline{K(U)}}({\bf C}_1\backslash \I{{\bf C}_1})\cup \cup_{i=1}^t{\V}_{\overline{K(U)}}({\P}\cup {\bf C}_1\cup \{\I{{C_1}_i}\})$. Let  ${\ARDU}({\P}, X)=[{\RS}, F]$ and  ${\tt{ZDToRC}}({\bf C}_1, X)=[{\RS}_1, F_1]$. Then ${\RS}_1\subset {\RS}$ and $F_1$ is a factor of $F$.  Let  ${\ARDU}({\P}\cup {\bf C}_1\cup\{\I{{C_1}_i}\},X)=[{{\RS}_2}_i, {F_2}_i]$ for every $i$ $(1\leq i\leq t)$. Then ${\RS}_1\cup \cup_{i=1}^t{{\RS}_2}_i={\RS}$ and $F=F_1\cdot \Pi_{i=1}^t{F_2}_i$ by Wu's method and Algorithm \ref{ALsus}.
We only prove the conclusion when ${\RS}_1\neq \emptyset$ and ${{\RS}_2}_i \neq \emptyset$ for every $i$ $(1\leq i\leq t)$. In fact, if ${\RS}_1=\emptyset$ or there exists $i$ $(1\leq i\leq t)$ such that ${{\RS}_2}_i=\emptyset$, the proof is similar.
For any $a\not\in {\V}^{U}(F)$, we know that $a\not\in V^{U}(F_1)$ and  $a\not\in V^{U}({F_2}_i)$ for every $i$ $(1\leq i\leq t)$. Hence for every $i$ $(1\leq i\leq t)$,
${\V}(({\P}\cup{\bf C}_1\cup \{\I{{C_1}_i}\})(a))=\cup_{{\T}\in {{\RS}_2}_i}{\V}({\T}(a))$ by the induction  hypothesis. Therefore, in order to prove ${\V}({\P}(a))=\cup_{{\T}\in{\RS}}{\V}({\T}(a))$, we only need to show
 ${\V}({\P}(a))=\cup_{\T\in{\RS}_1}{\V}({\T}(a))\cup \cup_{i=1}^t{\V}(({\P}\cup {\bf C}_1\cup \{\I{{C_1}_i}\})(a))$. Note that $\cup_{i=1}^t{\V}(({\P}\cup {\bf C}_1\cup \{\I{{C_1}_i}\})(a))={\V}(({\P}\cup {\bf C}_1\cup \{\I{{\bf C}_1}\})(a))$. So we only need to prove ${\V}({\P}(a))=\cup_{\T\in{\RS}_1}{\V}({\T}(a))\cup {\V}(({\P}\cup {\bf C}_1\cup \{\I{{\bf C}_1}\})(a))$. As a matter of fact, by Algorithm \ref{ALZDToRC}, $\cup_{\T\in{\RS}_1}{\V}({\T}(a))={\V}({\bf C}_1(a)\backslash \I{{\bf C}}_1(a))$. Note that ${\bf C}_1(a)\subset \ideal{{\P}(a)}_{\overline{K}[X]}$, so ${\V}({\P}(a))\subset {\V}({\bf C}_1(a))$. Then ${\V}({\P}(a))\subset {\V}({\bf C}_1(a)\backslash \I{{\bf C}}_1(a))\cup {\V}(({\P}\cup {\bf C}_1\cup \{\I{{\bf C}_1}\})(a))=\cup_{\T\in{\RS}_1}{\V}({\T}(a))\cup {\V}(({\P}\cup {\bf C}_1 \cup \{\I{{\bf C}_1}\})(a))$. On the other hand, by the claim (1),  ${\V}_{\overline{K(U)}}({\P})=\cup_{\T\in\RS}{\V}_{\overline{K(U)}}({\T})$ and thus for any ${\T}\in {\RS}_1\subset {\RS}$, ${\V}_{\overline{K(U)}}({\T})\subset {\V}_{\overline{K(U)}}({\P})$.  According to Algorithm \ref{ALZDToRC}, we know that $\res{\I{{\T}}, {\T}}(a)\neq 0$.  By Lemma \ref{LESPZERO}, $\cup_{{\T}\in{\RS}_1}{\V}({\T}(a))\cup {\V}(({\P}\cup{\bf C}_1\cup \{\I{{\bf C}_1}\})(a))\subset {\V}({\P}(a))$ and we are done.
\end{proof}

\begin{algorithm}\label{ALsus}
\SetAlgoCaptionSeparator{. }
\DontPrintSemicolon
\caption{\ARDU}
\KwIn{A generic zero-dimensional system ${\bf P}$ in $K[U][X]$, variables $X=\{x_1, \ldots, x_n\}$.}
\KwOut{[$\RS$, $F$], where\\

        (1)$\RS$ is a finite set of  zero-dimensional regular chains in $K[U][X]$ such that ${\V}_{\overline{K(U)}}({\P})=\cup_{{\T}\in {\RS}}{\V}_{\overline{K(U)}}({\T})$;

        (2)$F$ is a polynomial in $K[U]$ such that for any $a\in \overline{K}^d\backslash {\V}^{U}(F)$, ${\V}({\P}(a))=\cup_{{\T}\in {\RS}}{\V}({\T}(a))$ and ${\T}$ specializes well at $a$ for any $\T\in \RS$ if $ \RS\neq \emptyset$.}
     Compute a Wu's decomposition ${\mathbb S}=\{{\bf C}_1, \ldots, {\bf C}_m\}$ of $\P$ in $K[U][X]$ by Wu's method\;
     $\mathbb{S}_2$:=$\emptyset$,  $F$:=$1$\;
    \For {$i=1,\ldots,m$}{
     \eIf{${\bf C}_i$ is a contradictory ascending chain}
          { $F$:=$F\cdot \op{{\bf C}_i}$\;}
         {$\mathbb{S}_2$:=$\mathbb{S}_2\cup \{{\bf C}_i\}$\;}
     }
      \If{$\mathbb{S}_2=\emptyset$}{return [$\emptyset, F$]\;}
     $\RS$:=$\emptyset$\;
     \For{${\bf C}$ in $\mathbb{S}_2$}{
     W:=${\tt ZDToRC}({\bf C},X)$\;
      $\RS$:=${\RS}\cup W_1$,
      $F$:=$F\cdot W_2$\; }
     return [${\RS}, F$]\;
\end{algorithm}

Suppose $\V_{\overline{K(U)}}({\P})\neq \emptyset$ for a generic zero-dimensional system $\P$ in $K[U][X]$. Assume that $\mathbb{S}=\{{\bf C}_1, \ldots, {\bf C}_m\}$ is a Wu's decomposition of ${\bf P}\subset K[U][X]$ computed by Wu's method and let $\mathbb{W}=\{{\bf C}\in \mathbb{S}|\V_{\overline{K(U)}}({\bf C}\backslash\I{{\bf C}})\neq \emptyset\}$.
Note that $\mathbb{W}\neq \emptyset$ since $\V_{\overline{K(U)}}({\P})\neq \emptyset$. Obviously, ${\mathbb W}$ is a non-redundant
parametric triangular decomposition of $\P$ in $K[U][X]$. If $F={\ARDU}({\P},X)_2$,  it is obvious that the following corollary holds.
\begin{corollary}\label{CO}
The non-redundant parametric triangular decomposition $\mathbb{W}$ of $\P$ in $K[U][X]$ is stable at $a$ if $a\in \overline{K}^d\backslash {\V}^{U}(F)$.
\end{corollary}
Corollary \ref{CO} indicates that we can also obtain a non-redundant parametric triangular decomposition of a given generic zero-dimensional system and the decomposition is stable at any parameter value that is not on the RDU variety computed by Algorithm \ref{ALsus}.  Following the idea presented in Algorithm\ref{ALsus}, any algorithm for computing regular chain decompositions can be probably adopted to computing generic regular decompositions and the associated RDU varieties if the algorithm is based on resultants and pseudoremainders computation. The following Example \ref{EXintro} is presented to illustrate how Algorithm \ref{ALsus} and Corollary \ref{CO} work.
\begin{example}\label{EXintro}
Consider the system
 ${\P}=
 \{ (u-1)x_2^2+(x_1^2-2ux_1+u^2+1)x_2+x_1^2-x_1,
  (x_1-u)(x_2+1),
  (x_1-u)^2\}$,
where $x_1$ and $x_2$ are variables $(x_1\prec x_2)$ and $u$ is a parameter.
\end{example}
{\em\bf Step 1:} According to Wu's method, we compute a Wu's decomposition $\mathbb{S}=\{{\bf C}_1,{\bf C}_2,{\bf C}_3\}$ of $P$ in $\mathbb{R}[u][x_1,x_2]$ where ${\bf C}_1=\{(x_1-u)^2,(x_1-u)(x_2+1)\}$, ${\bf C}_2=\{x_1-u, (u-1)x_2^2+x_2+u^2-u\}$ and ${\bf C}_3=\{u-1\}$.

{\em\bf Step 2:} Let $\mathbb{S}_2=\{{\bf C}_1,{\bf C}_2\}$ and $F={\tt op}({\bf C}_3)=u-1$.

{\em\bf Step 3:} Because $\mathbb{S}_2\neq \emptyset$, we need to execute ${\tt ZDToRC}({\bf C}_1,\{x_1,x_2\})$ and ${\tt ZDToRC}({\bf C}_2,\{x_1,x_2\})$.

$\qquad$               {\em\bf Step 3.1:} It is easy to see that ${\bf C}_1$ is not a regular chain in ${\mathbb R}[u][x_1,x_2]$. By calling ${\APRSD}(\{(x_1-u)^2\},x_1-u, \{x_1\})$, we obtain $[\{\{x_1-u\}\}, \emptyset, 1]$. Then ${\tt ZDToRC}({\bf C}_1,\{x_1,x_2\})=[\emptyset, 1]$.

$\qquad$                {\em\bf Step 3.2:}  Since ${\bf C}_2$ is a regular chain in ${\mathbb R}[u][x_1,x_2]$,
${\tt ZDToRC}({\bf C}_2,\{x_1,x_2\})=[\{{\bf C}_2\},u-1]$.

Finally, we get a generic regular decomposition $\{{\bf C}_2\}$ of $\P$ in ${\mathbb R}[u][x_1,x_2]$ and a RDU variety $\mathcal{V}=\{a\in \mathbb{C}|a-1=0\}$ such that for any $a\in \mathbb{C}\backslash \mathcal{V}$, ${\V}(P(a))={\V}({\bf C}_2(a))$ where ${\bf C}_2(a)$ is a regular chain in $\mathbb{C}[x_1,x_2]$.
Furthermore, It should be noted that Algorithm \ref{ALsus} eliminates a redundant branch ${\bf C}_1$ from $\mathbb{S}_2$.  As a result, we also get a non-redundant parametric triangular decomposition $\{{\bf C}_2\}$ of $\P$.  As Corollary \ref{CO} shows, this non-redundant decomposition is stable at any $a\in \mathbb{C}\backslash \mathcal{V}$.

As introduced in Section \ref{SecIntro}, there exist several methods based on triangular sets for solving parametric systems. Now we present an example to compare the results computed by Algorithm \ref{ALsus} and functions {\tt Triangularize} and {\tt ComprehensiveTriangularize} in {\tt RegularChains}\footnotemark\footnotetext{Please see the help documents for these two functions in Maple 16.}.
\begin{example}\label{EXcompare}
Consider the system
${\P}=\{
   u_1x_2^2+x_1^2,
   u_1x_2^2+u_1x_1x_2+x_1\}$,
where $x_1$ and $x_2$ are variables $(x_1\prec x_2)$ and $u_1$ and $u_2$ are parameters $(u_1\prec u_2)$\footnotemark\footnotetext{The order of parameters is required when calling {\tt ComprehensiveTriangularize}.  }.
\end{example}
(1)By calling
{\tt Triangularize}$({\P}, {\tt PolynomialRing}([x_2,x_1],\{u_1,u_2\}))$,
we get a set $\{{\T}_1, {\T}_2\}$ of regular chains  in ${\mathbb R}[u_1,u_2][x_1,x_2]$
such that ${\V}_{\overline{{\mathbb R}(U)}}({\P})=\cup_{i=1}^2{\V}_{\overline{{\mathbb R}(U)}}({\bf T}_i\backslash \I{{\bf T}_i})$. Since ${\T}_1$ and ${\T}_2$ are zero-dimensional, ${\V}_{\overline{{\mathbb R}(U)}}({\P})=\cup_{i=1}^2{\V}_{\overline{{\mathbb R}(U)}}({\bf T}_i)$ by Lemma \ref{PRRCP1}.

(2)By calling ${\tt RDUForZD}({\P},\{x_1,x_2\})$, we get a set $\{{\T}_1,{\T_2}\}$ of regular chains in ${\mathbb R}[u_1,u_2][x_1,x_2]$ and $F=u_1u_2(u_1^3+u_2^2)$ such that ${\V}_{\overline{{\mathbb R}(U)}}({\P})=\cup_{i=1}^2{\V}_{\overline{{\mathbb R}(U)}}({\bf T}_i)$ and for any $a\in {\mathbb C}^2\backslash {\V}^{U}(F)$, ${\V}({\P}(a))=\cup_{i=1}^2{\V}({\bf T}_i(a))$ and ${\T}_i$ specializes well at $a$ for every $i$ $(1\leq i\leq 2)$.

(3)By calling
${\tt ComprehensiveTriangularize}({\P},2,{\tt PolynomialRing}[x_2,x_1,u_2,u_1])$,
we get five triples $({\mathbb T}_i, {\bf A}_i, {\bf B}_i)$ ($1\leq i\leq 5$) where ${\mathbb T}_i$ is a set of regular chains in ${\mathbb R}[u_1,u_2,x_1,x_2]$ and ${\bf A}_i$ and ${\bf B}_i$ are sets of polynomials in ${\mathbb R}[u_1,u_2]$, namely
$[(\{{\T}_1,{\T}_2\},\emptyset,\{u_1,u_2,u_1^3+u_2^2\}),
    (\{{\T}_1,{\T}_2,{\T}_3\},\{u_1\},\\\{u_2\}),
    (\{{\T}_2,{\T}_5\},\{u_2\},\\\{u_1\}),
    (\{{\T}_2,{\T}_3\},\{u_1,u_2\},\{1\}),
    (\{{\T}_2,{\T}_4\},\{u_1^3+u_2^2\},\{u_1\})]$,
such that ${\mathbb C}^2=\cup_{i=1}^5{\V}^U({\bf A}_i\backslash {\bf B}_i)$,
${\V}^U({\bf A}_i\backslash {\bf B}_i)\cap {\V}^U({\bf A}_j\backslash {\bf B}_j)=\emptyset$ for any $i\neq j$($1\leq i,j\leq 5$) and for any $a\in {\V}^{U}({\bf A}_i\backslash {\bf B}_i)$, ${\V}({\P}(a))={\V}_{{\T}\in {\mathbb T}_i}{\T}(a)$ and $\T$ specializes well at $a$ for any $\T\in {\mathbb T}_i$.

In the above presentation, ${\T}_1=\{x_1, x_2\}$, ${\T}_2=\{ (u_1^2+u_2^3)x_1^2+2u_1^2x_1+u_1, u_2x_2+u_1x_1+1\}$, ${\T}_3=\{u_1, x_1\}$, ${\T}_4=\{u_1^3+u_2^2, 2u_1x_1+1, u_2x_2+u_1x_1+1\}$ and ${\T}_5=\{u_2, u_1x_1+1, u_1x_2^2+x_1^2\}$.

Example \ref{EXcompare} shows that for the given generic zero-dimensional system,  the regular chains decomposition computed by {\tt Triangularize} is the same as the first output of Algorithm \ref{ALsus}. In addition, Algorithm \ref{ALsus} has a second output $F$ in ${\mathbb R}[U]$ so that for any $a\in {\mathbb C}^2\backslash {\V}^{U}(F)$, the regular chains decomposition is stable.  It is also indicated that  the result computed by Algorithm \ref{ALsus} is not as complete as that computed by {\tt ComprenhensiveTriangularize} since the latter gives a full answer
  to questions (1) and (2) proposed in Section \ref{SecIntro} and Algorithm \ref{ALsus} omits the parameter values on the affine variety generated by the second output $F$.

 \begin{remark}
 Actually, it is a further idea that we can compute comprehensive triangular decompositions by calling Algorithm \ref{ALsus} step by step. For instance,
consider the system ${\P}=\{u_1 x_{1}^{2}+u_2 x_{2}+1, u_2 x_{2}^{2}+x_1\}$.
 By calling ${\tt RDUForZD}({\P}, \{x_1,x_2\})$, we get  ${\mathbb T}_1=\{[u_{1}^{2} x_{1}^{4}+2 u_{1} x_1^2+u_2 x_1+1, u_2 x_2+u_{1} x_1^2+1]\}$ and the related RDU variety ${\mathcal V}_1=\{(u_1, u_2)\in {\mathbb C}^2|u_1u_2=0\}$ in ${\mathbb C}^2$. Then let ${\P}_1={\P}\cup \{u_1u_2\}$. By calling ${\tt RDUForZD}({\P}_1, \{u_1, x_1, x_2\})$, we get ${\mathbb T}_2=\{[u_1, u_2 x_1+1,
 u_2 x_2+1]\}$ and the related RDU variety  ${\mathcal V}_2=\{(u_1, u_2)\in V_1|u_2=0\}$.  Let ${\P}_2={\P}\cup \{u_1u_2, u_2\}$. Regard $u_2$ as a new variable.
By Algorithm \ref{ALsus}, we compute a regular chain decomposition $\{1\}$ of ${\P}_2$ in ${\mathbb R}[u_2, u_1, x_1, x_2]$. Therefore ${\P}(a)$ has no solutions in
${\mathbb C}$ for all $a\in {\mathcal V}_2$. Finally, we divide ${\mathbb C}^2$ into three parts: ${\mathbb C}^2\backslash {\mathcal V}_1$, ${\mathcal V}_1 \backslash {\mathcal V}_2$ and ${\mathcal V}_2$ and over each part, we have regular chains to represent the solutions of $\P$. To give a precise description for this method,  we need to deal with generic positive-dimensional systems and consider all the parameter values without overlapping.  We will discuss this issue in the future.
\end{remark}

\section{Implementation}\label{sectionexamples}
\begin{Table}\label{rsdtable}
      \begin{center}
      Comparing {\tt WRSD} and {\tt RSD}\par
\begin{tabular}{|c|c|c|c|c|c|c|c|c|c|c|}
     \hline
         & &&\multicolumn{2}{c|}{time}&\multicolumn{2}{c|}{$\mathbb{H}$}&\multicolumn{2}{c|}{$\mathbb{G}$}
       &\multicolumn{2}{c|}{$F$}\\
     \cline{4-11}
          system &$U$& $X$  &{\tt WRSD}& {\tt RSD}&{\tt WRSD}&{\tt RSD}&{\tt WRSD}&{\tt RSD}&{\tt WRSD}&{\tt RSD}\\
     \hline
     \hline
     {\em EX1} &3&7&0.406&0.718&1&1&2&2&11&11\\
     {\em EX2} &2&6&0.390&0.515&1&1&2&2&9&9\\
     {\em EX3}&4&6&1.295&2.793&1&1&2&2&10&10\\
     {\em EX4}&3&6&4.711&4.664&1&1&1&1&11&11\\
     {\em EX5}&4&4&0.780&0.765&1&1&1&1&9&9\\
     {\em EX6}&4&4&0.546&0.546&1&1&1&1&9&7\\
     {\em EX7}&3&3&0.842&1.045&1&1&1&1&8&8\\
     {\em EX8}&3&3&1.170&1.576&1&1&1&1&7&7\\
     \hline
    \end{tabular}
 \end{center}
\end{Table}
     We have implemented Algorithm \ref{ALsus} as a function $\tt{RDUForZD}$  on the basis of DISCOVERER  [22] using Maple 16. More specifically, Wu's method for computing parametric triangular decompositions introduced in Section \ref{sfuhaoshuoming} is implemented as a function ${\tt WUSOLVE}$ and Algorithm \ref{ALPRSD} is implemented as a function ${\tt WRSD}$.  Remark that we use factorization without loss of correctness when implementing. The details are omitted.
Throughout this section,  all the results are obtained in Maple 16 using an Intel(R) Core(TM) 2 Solo processor(1. 40GHz CPU and 2GB total memory).
\begin{Table}\label{timetable}
\begin{center}
     {Comparing {\tt RDUForZD} and \tt{TRIANGULARIZE}}\par
     \begin{tabular}{|c|c|c|c|c|c|c|c|}
     \hline
      number&system &$U$&$X$&$\tt{WUSOLVE}$ & $\tt{ZDToRC}$&\tt{RDUForZD}&\tt{TRIANGULARIZE}\\
     \hline
     \hline
      1. &{\em S5}  &4&4& 0.047& 0. &0.047&0.171\\
      2. &{\em S9}  &3&3& 0.078& 0.012&0.090&0.124\\
      3. &{\em S16} &3&12& 0.078& 0.&0.078&0.530\\
      4. &{\em S18}&3&2&1.388&0.094&1.482&2.200\\
      5. &{\em SY1}&3&2&0.063&0.&0.063&0.124\\
      6. &{\em SY2} &4&1&0.047&0.&0.047&0.046\\
      7. &{\em SCC1} &3&4 &0.016&0. &0.016&0.047\\
      8. &{\em SCC2}&4&7&0.047&0.015&0.062&0.156\\
      9. &{\em SCC3}&6&11&0.312&0.281&0.593&1.217\\
      10. &{\em SCC4} &4&7 &0.172&0.093&0.265&0.406\\
      11. &{\em SCC5} &4&5&0.078&0.016&0.094&0.141\\
      12. &{\em P3P} &5&2& 0.062 & 0.016&0.078&0.078\\
      13. &{\em F4} &4&1&0.016&0.&0.016&0.063\\
      14. &{\em F6} &4&1&0.031&0.&0.031&0.047\\
      15. &{\em Gerdt}&3&4&0.078&0.&0.078&0.078\\
      16. &{\em Wang93}&2&3&0.234&0.&0.234&0.624\\
      17. &{\em Leykin-1}&4&4&0.171&0.&0.171&0.219\\
      18. &{\em Neural}&1&3&0.218&0.016&0.234&0.281\\
      19. &{\em Pavelle}&4&4&0.686&0.203&0.889&0.811\\
      20. &{\em genLinSyst-3-3}&12&3&0.062&0.&0.062&0.140\\
      21. &{\em AlkashiSinus}&3&6&0.078&0.&0.078&0.141\\
      22. &{\em LanconeLLi}&7&4&0.219&0.&0.219&0.218\\
     \hline
    \end{tabular}
\end{center}
   \end{Table}
We run 8 examples\footnotemark\footnotetext{http://www.is.pku.edu.cn/\~{}xbc/ExForRSD.txt} using $\tt{WRSD}$ and $\tt{RSD}$\footnotemark\footnotetext{ Algorithm {\rm RSD} in \cite{zjzi} was implemented as a subfunction {\tt RSD} in DISCOVERER by Xia \cite{discover}.} on the same computer with Maple 16 and the comparisons about timings and results are presented in Table \ref{rsdtable}, where columns $X$ and $U$ represent the cardinal numbers of the variables and the parameters, respectively, column $time$ reports the timings in seconds,  columns $\mathbb{H}$ and $\mathbb{G}$ represent the numbers of branches in the first and second outputs, respectively, column $F$ represents the numbers of irreducible factors over the field of rational numbers of the third output. The empirical data shows that {\tt WRSD} performs as well as {\tt RSD} with higher efficiency in most cases.

We also run several examples from the literature \cite{changbo, zxq, sun, Montes} using {\tt RDUForZD} with Maple 16 and part of the empirical data about timings is presented in Table \ref{timetable}.  In Table \ref{timetable},
       column ${\tt WUSOLVE}$ reports the timings in seconds cost by computing Wu's decompositions,
       column ${\tt ZDToRC}$ represents the timings in seconds cost by computing weakly relatively simplicial decompositions and some other steps required in Algorithm \ref{ALsus}, column $\tt{RDUForZD}$ shows the timings added by the timings in the former two columns and column $\tt{Triangularize}$ shows the timings in seconds cost by the function {\tt Triangularize}\footnotemark\footnotetext{For a given test-system $\P$, we call ${\tt Triangularize}({\P}, {\tt PolynomialRing}([x_n,\ldots,x_1],\{u_1,\ldots,u_d\}))$.} in Maple 16. It is indicated that our method can be applied to a wide range of practical problems with reasonable time cost. Furthermore, it is interesting to note that computing generic regular decompositions and the associated RDU varieties do not require much more time cost than Wu's decompositions when solving practical problems as shown in Table \ref{timetable}.
\section{Conclusions}\label{sectionconclusions}
   We give an algorithm for computing generic regular decompositions and the associated RDU varieties simultaneously for generic zero-dimensional systems in this paper.   As a result, questions (1) and (2) in Section \ref{SecIntro} are answered to some extent.  In the future, we will discuss how to modify Algorithm \ref{ALsus} for general systems and then we will answer questions (1) and (2)  completely.  A clearer characterization of the relationship between BPs and RDU varieties is also an interesting topic for our future work.
\section*{Acknowledgements}
The work is partly supported by the National Natural Science Foundation of China (Grant No.11271034),  the ANR-NSFC project EXACTA (ANR-09-BLAN-0371-01/60911130369)
and the project SYSKF1207 from State Key Laboratory of Computer Science, Institute of Software, the Chinese Academy of Sciences. We would like to thank Changbo  Chen and Yao Sun for providing a great deal of test-systems. Thanks also go to Rong Xiao for his suggestions. We would like to thank Hoon Hong and Dongming Wang for their advices on the original version of this paper.  Thanks also go to the reviewers for their valuable comments.

\end{document}